\documentclass[final,12pt]{elsarticle}
\usepackage{a4, amsmath, amssymb, graphics, subfigure, fullpage, color}
\usepackage[colorlinks]{hyperref}
\usepackage{ulem} 
\hypersetup{colorlinks,citecolor=green,filecolor=black,linkcolor=blue,urlcolor=blue}
\usepackage{fancyhdr}
\usepackage{multirow}
\usepackage{amsthm}
\usepackage{xcolor}

\newtheorem{theorem}{Theorem}
\newtheorem{definition}{Definition}
\newtheorem{corollary}{Corollary}

\newtheorem{lemma}{Lemma}

\newtheorem{remark}{Remark}
\newtheorem{assumption}{Assumption}
 \usepackage{soul}

\newcommand{\cA}{\mathcal{A}}

\newcommand{\FF}{\mathbb{F}}

\DeclareMathOperator{\E}{\mathbb{E}}

\renewcommand{\Pr}{\mathbb{P}}

\DeclareMathOperator{\var}{\mathbb{V}}

\usepackage{dsfont}
\DeclareMathOperator{\ind}{\mathds{1}}

\newcommand{\beq}{\begin{equation}}
\newcommand{\eeq}{\end{equation}}
\newcommand{\beqn}{\begin{eqnarray}}
\newcommand{\eeqn}{\end{eqnarray}}

\usepackage{lineno}
\usepackage[utf8]{inputenc}
\usepackage[T1]{fontenc}
\usepackage{bbm}
\usepackage{natbib}
\newcommand{\MyBib}{./MyBib}

\journal{Statistics and Probability Letters}

\begin{document}

\begin{frontmatter}

\title{Exact calibration of structural models via time-change\tnoteref{myfn}}


\author[FV]{Fr\'ed\'eric VRINS}
\address[FV]{UCLouvain, LIDAM/LFIN, Belgium\\ \href{mailto:frederic.vrins@uclouvain.be}{frederic.vrins@uclouvain.be}}
\author[DB]{Damiano BRIGO}
\address[DB]{Imperial College London, Department of Mathematics, United Kingdom\\ \href{mailto:damiano.brigo@imperial.ac.uk}{damiano.brigo@imperial.ac.uk}}
\tnotetext[myfn]{We thank Ludger Overbeck for drawing our attention, after the first posting of
this note, to Overbeck and Schmidt~\cite{OverbeckSchmidt2005}. We had arrived
independently at the Brownian time-change construction; we have revised the
paper to acknowledge that its marginal-calibration result was obtained
previously and to clarify the distinct emphasis of the present note.
}

\begin{abstract}
In this note, we propose a general structural approach to model a default time $\tau$ as the first-passage time (FPT) of a (``firm-value'') process $S$ below a (``debt'') barrier $K$ that comply with a pre-specified survival probability curve $G(t)=\Pr(\tau>t)$. Following an idea of Mbaye \& Vrins applied to reduced-form models, our approach consists in two steps: choose a latent FPT model driven by a barrier $\tilde{K}$ and process $\tilde{S}$, and time-change those using a deterministic clock $\Theta$ to get $K_t=\tilde{K}_{\Theta(t)}$ and $S_t:=\tilde{S}_{\Theta(t)}$, leading to the final FTP model $(K,S,\Theta)$. As the market curve $G$ and the latent model $(\tilde{K},\tilde{S})$ are assumed to be given, the calibration step simply consists in finding the clock $\Theta$ such that the distribution of the FPT of $S$ below $K$ coincides with the survival curve $G$. We show that this is achievable for a broad class of specified curves $G$ and latent FTP models. The calibration amounts to a simple inversion of a function, which is almost immediate provided that the latent model is tractable enough. In particular, we show that the AT1P model of Brigo et al. (2004, 2005, 2011) \--- which is able to reproduce a broad range of CDS term-structures \--- can be regarded as the FPT of a time-changed drifted Brownian motion to a constant barrier: $\tilde{V}_t=\mu t+W_t$ and $\tilde{K}_t=k<0$. This connection offers an elegant interpretation for the instantaneous volatility function featured in AT1P and yields an immediate calibration of the latter to perfectly match a target survival curve.
\end{abstract}


\end{frontmatter}


\section{Introduction}

Structural default models consist in defining the default time $\tau$ of a reference entity as the first passage time (FPT) of a process $S=(S_t)_{t\geq 0}$ below a threshold, possibly time-dependent process $K=(K_t)_{t\geq 0}$ called \textit{barrier}. To ease the exposition, we assume here that $K_t=k(t)$ is a deterministic function, possibly constant, although the stochastic case can be handled, too.

The most popular model is probably Merton's approach. It does not model the default time $\tau$ but focuses instead on a default event $\{\tau<T\}$. The latter is modeled as $\{S_T<k\}$ where  $T$ is a fixed horizon, $S$ is a geometric Brownian motion (GBM) with constant parameters $(\mu,\sigma)$ representing the value of the firm, and $k(t)=k$ is a constant threshold interpreted as the debt level~\cite{Merton74}. The Black-Cox extension, however, explicitly defines $\tau$ as the FPT of the GBM $S$ below $k$~\cite{Black76}. We follow this general mathematical setup here:
\beq
\tau:=\inf\{t:S_t\leq k(t)\}=\inf\{t:S_t-k(t)\leq 0\}\label{eq:tau}
\eeq
for some firm-value $S$ and barrier $k$ processes. Unless stated otherwise, we assume in this paper that $\tau>0$, i.e., $S_0> k(0)$. \medskip

Structural models receive specific attention in the credit risk literature because of their sound economic interpretation. This contrast with intensity (reduced-form) models, for instance. Of course, the standard models recalled above depict the firm-value and the level of debt in a very naive way, to say the least. More economically relevant processes can be considered (e.g., where $S$ features jumps), but this comes at the cost of a reduced analytical tractability when it comes to compute the distribution of the FPT $\tau$ in \eqref{eq:tau}. This highlights the usual compromise between economic plausibility and mathematical tractability in this specific context. Note that, in practice, this specific financial interpretation has to be tempered, as the value of the assets of firms is not observable anyway. A second advantage of structural models has to be found in their ability to generate high credit-credit and market-credit correlations, such as in credit valuation adjustment (CVA) under wrong-way risk (WWR) (see for example \cite{Brigo13}, Chapter 8.4 for equity portolios).

However, a well-known drawback of structural models is their inability to fit market curves. In practice, indeed, the model needs to be calibrated to the market, in the sense that the parameters of the stochastic process $S$ and/or the shape of the barrier $k$ must be chosen such that the associated survival distribution $G^\tau(t):=\Pr(\tau>t)$ agrees with a given (typically market-implied) target survival probability curve $G:\mathbb{R}_+\mapsto(0,1]\;,~t\to G(t)$. The calibration procedure consists of choosing the dynamics of $S$ and the shape of the barrier $k$ such that the probability that $S$ stays above the barrier up to $t$ is given by $G(t)$:
\beq
G^\tau(t):=\Pr\left(\inf_{s\in[0,t]} (S_s-k(s))> 0 \right)=G(t)\;,\quad\forall t\geq 0\;.\label{eq:cal}
\eeq

Computing $G^\tau$ when the default time $\tau$ takes the form \eqref{eq:tau} is a difficult problem, unless the dynamics of the process $S$ and the shape of the barrier $k$ are ``simple enough''. In those cases, however, the model only features very few parameters, and there is \textit{a priori} no reason for the target curve $G$ to belong to the set of parametric curves $G^\tau$ that can be generated by the model at hand.\footnote{In the case of Black-Cox, for instance, the model features 4 parameters $(S_0,\mu,\sigma,k)$ but only three degrees of freedom, because the pair $(S_0,k)$ only shows up as the ratio.} \medskip

Let us leave the economic interpretation of structural models aside to focus on the analytically tractable setups. Perhaps the simplest case is to take as $S$ a Brownian motion $W$ (or geometric Brownian motion with constant parameters) and a linear barrier $k(t)=at+b$ (or an exponential deterministic barrier):

\[
G^\tau(t):=\Pr\left(\inf_{s\in[0,t]} -as+W_s >b \right)\;,\quad\forall t\geq 0\;.
\]
This amounts to study the PFT of a Brownian motion with constant drift $\mu=-a$ below a constant barrier $k(t)=b$, whose analytical expression is well-known; see \eqref{eq:Mtidle} below.\medskip

However, as explained above, those cases are of little practical interest because of their very limited flexibility. Indeed, the number of degrees of freedom is bounded by the number of parameters. This suggests that, unless $G$ has a very particular parametric form, there is no pair $(a,b)$ such that equality $G^\tau(t)=G(t)$ holds at all times $t\geq 0$. One possibility to force the calibration $G^\tau=G$ would be to have a functional parameter, e.g., fix the parameters of $S$, and infer the shape of the deterministic function $k:t\to k(t)$ accordingly. Unfortunately, as explained above, this is not possible because the distribution of the running infimum of $S_s-k(s)$ is known only for very special cases (see~\cite{JinWang17} for a recent review of existing results). \medskip 

To circumvent this problem, we propose a different, very simple route. We start by considering the case of the passage time of the running infimum $\underline{X}$ of a continuous stochastic process $X$ below a constant $k$:
\beq\label{eq:DefModel}
\tau^X:=\inf\{t: X_t\leq k\}=\inf\{t:\underline{X}_t=k\}\quad\text{where}\quad k<X_0\quad\text{and}\quad \underline{X}_t:=\inf_{s\in[0,t]} X_s.
\eeq
We take $X$ to be ``simple enough'' in order for the probability of the events $\{\underline{X}_t>k\}$ to be known in (semi-)closed form,
\beq
G^{\tau^X}(t)=\Pr\left(\underline{X}_t> k\right)=:G^{X}_k(t).\label{eq:Ginf}
\eeq

Next, we consider $X$ as a latent process and model the firm-value as $S= X^\theta$ where $X^\theta$ is a time-changed version of $X$, for some appropriate function $\Theta$ called \textit{clock}. Time-changed approaches have been widely used in quantitative finance, and in credit risk in particular. For instance, time-changing a stochastic process having continuous trajectories with a \textit{subordinator} that exhibits jumps is a very convenient way to introduce jumps in the dynamics of $S$, while preserving semi-analytical properties; see, for instance, \cite{Cont04,Hurd09} for application in credit risk, and \cite{Prof19} for the special case of lazy clocks. Such models help to design more realistic structural models featuring a ``surprise effect'', leading to higher implied short-term spreads. In this note, we are using the time-change technique not to modify the type of sample paths of $S$ in a tractable way, but to enrich the dynamics of  a simple parametric model with the required flexibility in order to match a given target curve $G$.\medskip

In this setup, the dynamics of the latent process $X$ and the (constant) barrier $k$ are assumed to be given, and the calibration problem amounts to finding the implied clock $\Theta$ such that the FPT of $S=X^\theta=(X_{\Theta(t)})_{t\geq 0}$ below $k<0$ has a survival probability function given by the target curve $G$. In other words, instead of adjusting the barrier function $k(t)$ to $G$, we fix $k(t)=k$ and adjust the clock $\Theta$, assuming that such a clock exists. A similar approach was recently introduced by Mbaye and Vrins in \cite{Vrins22} for reduced-form models, as an alternative to the shifted square-root jump diffusion. Time-changing a non-negative intensity process with a clock $\Theta$ instead of shifting it with a deterministic function $\varphi$ solves the negative intensities issue impacting the popular ``deterministic shift'' approach (J)CIR++. This is because the implied shift function $\varphi$ is typically not positive, while time-changing a non-negative process preserves its non-negativity.
\medskip

Deterministic time changes of Wiener threshold models have also been used
for the calibration of prescribed marginal default distributions. In
particular, Overbeck and Schmidt~\cite{OverbeckSchmidt2005} derive the
explicit Brownian time-change formula that matches a given default curve and
combine correlated Wiener drivers with these time changes to address joint
default probabilities and basket-credit applications. We consider this Brownian case as a particular example illustrating our general approach. The clock in Theorem~1 below coincides with their construction after writing
the target default distribution as $F=1-G$. 

The present note builds on this perspective in three directions. First, we show that the time-change idea can be applied not only to calibrate a model relying on a plain Brownian motion with fixed barrier ($X=W, k(t)=k<0)$, but can be extended to a broad class of latent models called \emph{regular} hereafter. We provide sufficient conditions on the target curve and the latent model for a unique implied clock to exist, which is given by
the inverse-survival transformation
\[
\Theta_k^{G,X}=(G_k^X)^{-1}\circ G.
\]
Notably, this includes survival curves with
piecewise constant hazard rates and encompasses some non-diffusive latent models. Second, it identifies the AT1P survival
formula as that of a drifted Brownian first-passage model evaluated at its
integrated variance, thereby interpreting the AT1P variance profile as a
deterministic clock and providing a direct calibration to a prescribed
survival curve. Third, we show that the sharp increase of the clock near the origin highlighted empirically in~\cite{OverbeckSchmidt2005} is actually a singularity: in the Brownian case with fixed barrier, the clock rate explodes ar the origin when the survival curve $G$ is strictly decreasing there. This singularity in the volatility raises questions the validity of the approach. We verify that, in the Brownian case, the singular variance rate associated with the implied clock remains
locally integrable for every regular target survival curve. These results are
useful in particular for relating the time-change construction to the
AT1P/SBTV structural-credit modelling line and to joint credit--market
applications, where the clock also affects instantaneous variance and
cross-variation with other risk factors. Although the proposal to choose the fixed barrier level $k$ so as to make sure the calendar and clock times agree at $t_0$ prevents the clock to deviate too much from identity, this does not   resolve the exploding behavior at the origin if for example  $h(0+):=-(G'/G)(0^+)>0$, as the latter is inherent to the joint behavior of $G,G^X$ at the operational time $0$.
Time-change techniques are thus very promising for calibration purposes in that they allow to ``stretch time'' in a very flexible way. Unfortunately, these approaches also trigger some extra difficulties. For instance, it is not obvious at first sight that there always exists such an implied clock $\Theta$ and, if so, how it can be found. Second, dealing with several time-changed processes at a same time is difficult when the clocks differ. This is for example the case when several dynamic models need to be correlated but have different clocks, as in the WWR CVA example; as time moves differently for the various processes, managing the multiple filtrations is far from handy.

In order to build default models that are practically appealing and easier to deal with, it is important to have a setup framed in the original time scale. To achieve this goal, we rely on a representation theorem of time-homogeneous diffusions as time-changed Brownian motions. This allows us to conceive the default time $\tau$ not as the FPT of a time-changed process $X^\theta$, but as the FPT of a time-homogeneous martingale adapted to the original filtration whose deterministic diffusion coefficient $\sigma(\cdot)$ is determined according to $\langle X^\theta\rangle$, the quadratic variation $X^\theta$.\medskip

In this note, we start by considering the case where the latent process $X=W$ is a Brownian motion. We show that, in this case, the implied clock $\Theta=\Theta^{G,W}_k$ is available in closed form as a function of the threshold $k$ and the function $G$. Next, we explain that our method is actually quite universal, in the sense that it can be applied to a broad range of FPT models. We consider various extensions, where our time-changed approach is used to calibrate Black-Cox or the AT1P model of Brigo et al.~\cite{Brigo13,brigomorinitarenghi,Brigo2004,Brigo2005}. In fact, we show that the latter model is exactly the FPT of a  time-changed Brownian motion with constant drift: one way to calibrate the AT1P model to a given survival curve $G$ is to set $\sigma^2(t)=\theta(t)$ where $\Theta$ is the implied clock. This interpretation of the A1TP model has three key benefits. First, it offers an almost immediate calibration of the AT1P model, as the clock is obtained by simple numerical inversion of a known, continuous curve.\footnote{The calibration considered in~\cite{Brigo13,brigomorinitarenghi,Brigo2004} postulate a piecewise-constant volatility function $\sigma(t)$ whose levels $\sigma_1,\sigma_2,\ldots,\sigma_n$ are estimated iteratively based on a term structure of CDS spreads featuring $n$ tenors, by reverse-engineering the standard CDS pricing formula. It is clear that if $G$ is extracted from such a CDS term structure, our calibration would yield a different volatility function, not necessarily piecewise constant, but our approach would still yield an AT1P model that successfully reprices all the $n$ CDS contracts.} Second, in contrast with the calibration technique proposed in~\cite{Brigo13,brigomorinitarenghi,Brigo2004}, our approach allows to calibrate the AT1P model not just on a finite set of $n$ CDS quotes, but to perfectly replicate a target survival probability curve at all times, such as the piecewise constant hazard rate parametrization associated with the ISDA model, in force in Bloomberg{\tiny\copyright}.\footnote{Note that by ``perfect replication'', we simply mean that the model displays the right survival curve at time 0, i.e., that the survival probability curve $\Pr(\tau>t)$ implied by the model matches the survival curve $G(t)$ extracted from market quotes for all $t\geq 0$. In particular, we stay silent about the future, conditional survival probabilities.} Finally, it is clear from our theory that the calibration approach we propose always lead to a successful calibration. To the best of our knowledge, this is not guaranteed when adopting the calibration method introduced in~\cite{Brigo13,brigomorinitarenghi,Brigo2004}, featuring a piecewise constant volatility. 
\section{Mathematical setup}

The survival probability implied by the model up to time $t$ for a given set of parameters (i.e., the parameters of $S$) is given by \eqref{eq:Ginf}. In this paper, we consider a firm-value process of the form $S=X^\theta$, i.e., as a time-changed transform of a latent process $X$, so that the no-default event can be given by the following equivalent conditions:

\[
\{\tau>t\}\Leftrightarrow \left\{\inf_{s\in[0,t]} S_s> k\right\}\Leftrightarrow \left\{\inf_{s\in[0,t]} X^\theta_{s}> k\right\}\Leftrightarrow \left\{\inf_{s\in[0,t]} X_{\Theta(s)}> k\right\}.
\]

In this paper, we restrict ourselves to consider a specific set of time-changed transforms called \textit{clocks}; see e.g. \cite[Definition 4]{Vrins22}.\footnote{A clock is usually defined as a non-decreasing process and is not restricted to be continuous. The general definition of a clock actually coincides with the concept of \textit{subordinator}. In particular, pure jump processes such as Variance-Gamma can be used, too. In this paper, however, we assume a clock is deterministic and continuous. See e.g. \cite{Cont04}, \cite{Hurd09} and \cite{Prof19}.} We assume in the sequel that $X_0>k$.

\begin{definition}[Clock]\label{def:clock}
A \emph{clock} is an absolutely continuous function
$\Theta:\mathbb{R}_+\to\mathbb{R}_+$ such that $\Theta(0)=0$, $\Theta$ is
strictly increasing, and $\lim_{t\to\infty}\Theta(t)=+\infty$. Equivalently,
\[
\Theta(t)=\int_0^t\theta(s)\,ds,
\]
where $\theta\in L^1_{\mathrm{loc}}(\mathbb{R}_+)$, $\theta(s)>0$ for almost
every $s>0$, and $\int_0^\infty\theta(s)\,ds=+\infty$. We call $\theta$ the
\emph{clock rate}; it is understood only up to equality almost everywhere.
\end{definition}

Thus a clock needs not be differentiable at every time: it is differentiable
almost everywhere, with $\Theta'(t)=\theta(t)$ almost everywhere. In
particular, piecewise linear clocks, and hence clocks generated by piecewise
constant clock rates, are covered by this definition. The next Lemma requires only
that $\Theta$ be continuous and strictly increasing.

\begin{lemma}
Let $\Theta$ be a clock in the sense of Definition \ref{def:clock}. Then, 
$\underline{X^\theta}_t=\underline{X}_{\Theta(t)}$.
\end{lemma}
\begin{proof}
The proof is straightforward. Indeed,

\beq
\underline{X^\theta}_t=\inf_{s\in [0,t]} X^\theta_s=\inf_{s\in [0,t]} X_{\Theta(s)}=\inf_{s\in [0,\Theta(t)]} X_s=\underline{X}_{\Theta(t)}\;,\label{eq:TCsup}
\eeq
where the first, second and last equalities result from the definition of $\underline{X^\theta}, X^\theta$ and $\underline{X}$, respectively. The third equality results from the invertibility of the (continuous) map $t\mapsto \Theta(t)$.
\end{proof}

In our setup, the process $X$, the constant barrier level $k$ and the target curve $G$ are given. The calibration problem thus simply consists in finding the clock $\Theta$ such that the distribution function of the time-$t$ running infimum of $S=X^\theta$ at $k$ agrees with the curve $G$ given exogenously:
$$G(t)=G^\tau(t)=G^{S}_k(t)=\Pr(\underline{S}_t> k)=\Pr(\underline{X^\theta}_{t}> k)=\Pr(\underline{X}_{\Theta(t)}> k)=G^{X}_k(\Theta(t))\;,$$
where $G^X_k(t):=\Pr(\underline{X}_t>k)$. In other words, the calibration problem amounts to find a specific $\Theta$ called \textit{implied clock}.

\begin{definition}[Implied clock]\label{def:iclock} The \textit{clock implied by $(G,X,k)$} is the time-change function $\Theta^{G,X}_k$ satisfying $G^X_k\left(\Theta^{G,X}_k(t)\right)=G(t)$ for all $t\geq 0$.
\end{definition}

It is easy to check that the implied clock exists, is unique, and is indeed a clock in the sense of Definition~\ref{def:clock} under relatively mild assumptions on $X,G$ whenever $k<X_0$.

\begin{definition}[Regular survival function]\label{def:RegularSP}
A survival function $G:\mathbb{R}_+\to(0,1]$ is called \emph{regular} if
\begin{itemize}
\item $G(0)=1$;
\item $\lim_{t\to\infty}G(t)=0$;
\item $G$ is absolutely continuous on every compact interval and
$G'(t)<0$ for almost every $t>0$.
\end{itemize}
\end{definition}

For a regular survival function, define, almost everywhere,
\[
h(t):=-\frac{G'(t)}{G(t)}>0,
\qquad
H(t):=\int_0^t h(s)\,ds.
\]
Then $G(t)=\exp(-H(t))$ and $g(t):=-G'(t)=h(t)G(t)$ almost everywhere.

The definition permits hazard rates that become arbitrarily small near the origin; no pointwise value $h(0)$ is assumed. Our purpose is not to prescribe which latent process $X$ should represent firm value. Rather, because the method applies to a broad class of first-passage models, we take a latent process $X$ as given and calibrate it to a target curve $G$ by an appropriate time change. Thus, the only required compatibility condition is that $G$ and the latent survival curve $G_k^X$ share the regularity specified below.

\begin{assumption}[$\mathcal{A}_1$]\label{ass:A1}
Both the market survival curve $G$ and the latent survival curve $G_k^X$ are
regular survival functions in the sense of Definition~\ref{def:RegularSP}.
\end{assumption}

The regularity property (of a survival curve) is invariant under precomposition with a clock.  

\begin{lemma}[Invariance under precomposition] Let $G$ be a regular survival function and $\Theta$ be a clock. Then, $G\circ\Theta$ is regular, too. 
\end{lemma}
\begin{proof}
Let $G^\Theta(t):=G(\Theta(t))$. The endpoint properties follow from
$\Theta(0)=0$ and $\Theta(t)\to\infty$. Since both $G$ and $\Theta$ are
absolutely continuous on compact intervals, so is $G^\Theta$, and the
chain rule for absolutely continuous functions gives
\[
(G^\Theta)'(t)=G'(\Theta(t))\theta(t)
\quad\text{for almost every }t>0.
\]
This derivative is strictly negative almost everywhere. Hence $G^\Theta$ is
regular.
\end{proof}

In particular, under Assumption~\ref{ass:A1}, the compounded map $t\to G^{X^\theta}_k(t):=G^X_k\circ\Theta(t)$ is a regular survival function.

\begin{lemma}\label{lem:iclock}
Under Assumption~\ref{ass:A1}, the calibration problem admits a unique
implied clock $\Theta_k^{G,X}$. If $I_k^X:(0,1]\to[0,\infty)$ denotes the
inverse of the strictly decreasing function $G_k^X$, then
\[
\Theta_k^{G,X}(t)=I_k^X(G(t)),\qquad t\geq0.
\]
Moreover, with $g:=-G'$ and $g_k^X:=-(G_k^X)'$, one has, for almost every
$t>0$,
\[
(\Theta_k^{G,X})'(t)
=\frac{g(t)}{g_k^X(\Theta_k^{G,X}(t))},
\qquad \Theta_k^{G,X}(0)=0.
\]
\end{lemma}
\begin{proof}
By Assumption~\ref{ass:A1}, $G$ and $G_k^X$ are continuous strictly decreasing bijections from $[0,\infty)$ onto $(0,1]$, with the stated endpoint limits. Thus $I_k^X$ is well defined and $\Theta(t):=I_k^X(G(t))$ is the unique function satisfying $G_k^X(\Theta(t))=G(t)$. Standard inverse and chain rules for strictly monotone absolutely continuous functions yield that $\Theta$ is absolutely continuous on compact intervals and give the displayed derivative identity almost everywhere. Its derivative is positive almost everywhere; its endpoint properties follow from those of $G$ and $G_k^X$. Hence $\Theta$ is a clock in the sense of Definition~\ref{def:clock}, and it is the implied clock.
\end{proof}

\section{A first approach using a Brownian motion}\label{sec:FPT-BM}

As a guiding example, let us start with the case where $X=W$ where $W$ is a Brownian motion. The law of the running infimum $\underline{W}_t$ can be found by using the reflection principle \cite[Corollary 3.4]{Baldi17} :

\beq
G^{W}_k(t)=\Pr(\underline{W}_t> k)=2\Phi\left(\frac{-k}{\sqrt{t}}\right)-1\;,\quad\forall k< 0.\label{eq:RMBM}
\eeq
Observe that the derivative of this function, i.e., $-g^{W}_k(t)$, vanishes at $t=0$. Yet, $G^W_k$ is regular in the sense of Definition~\ref{def:RegularSP}. This function has a single parameter ($k$), which is not enough to fit an arbitrary target survival curve $G$, in general.\medskip

Let us now construct a model where the default time $\tau^\theta$ is defined as the FPT of \textit{a time-changed Brownian motion} $S=X^\theta$ below $k<X_0=0$. We show that, thanks to the flexibility provided by the time change function, it is possible to make the survival probability function generated by the model agree with any given survival function $G$, at any time $t>0$, provided that Assumption \ref{ass:A1} is met.\footnote{
For $X=W$, formula~\eqref{eq:ThetaCal} is the Brownian deterministic-time-change
construction of Overbeck and Schmidt~\cite{OverbeckSchmidt2005}, expressed
in terms of the survival function $G$ rather than the default distribution
$F=1-G$.
}

\begin{theorem}\label{th:clock}
Let $G$ be a regular survival function and set $X=W$. Then, for each $k<0$, the implied clock is $\Theta:t\mapsto\Theta(t)$ is given by 
\beq\label{eq:ThetaCal}
\Theta^{G,W}_k(t):=\left(\frac{k}{\Phi^{-1}\left(\frac{1+G(t)}{2}\right)}\right)^2\;.
\eeq
\end{theorem}
\begin{proof}
By Assumption~\ref{ass:A1}, $G$ is regular; the Brownian latent survival curve
$G_k^W$ is regular for $k<0$. Hence Lemma~\ref{lem:iclock} implies that the function
$\Theta_k^{G,W}$ in~\eqref{eq:ThetaCal}, which is the inverse representation
$(G_k^W)^{-1}\circ G$, is a clock in the sense of Definition~\ref{def:clock}. Therefore,
using~\eqref{eq:TCsup},
\beq
\Pr(\tau^\theta>t)=\Pr(\underline{W^\theta}_t> k)=\Pr(\underline{W}_{\Theta(t)}> k)=2\Phi\left(\frac{-k}{\sqrt{\Theta(t)}}\right)-1=G(t)\;.
\eeq
\end{proof}

The above theorem shows that it is actually very easy to build a structural default model that can be perfectly calibrated to a given initial survival function $G$ that is regular: it suffices to define the default time as $\tau^\theta$, the FPT of a time-changed Brownian motion $W^\theta$  below a constant level $k<0$, where the clock $\Theta$ is set according to \eqref{eq:ThetaCal}. Time-change techniques are thus very powerful in that they allow to ``stretch time'' in the appropriate way. Moreover, the implied clock $\Theta^{G,W}_k$ is available in closed form once the threshold $k$ and the curve $G$ are provided: for any threshold $k<0$ and regular survival function $G$, that clock $\Theta$ ensures $G^{\tau^\theta}(t)=G(t)$ for all $t\geq 0$. In other words, the model can be calibrated automatically. To avoid the challenges of handling several clocks (hence, filtrations) at a same time, we recast our framework in the original time scale.\medskip 

The next corollary shows that $\tau^\theta$ can be regarded as the default time of a structural model expressed in the natural time-scale, i.e., as the first passage time $\tau^S$ of an $\FF$-adapted process $S$ below the threshold $k$. The corresponding firm-value process $S$ is a diffusion martingale with deterministic diffusion coefficient $\sigma$ which is determined so as to fit the quadratic variation of $W^\theta$.\medskip

\begin{corollary}\label{cor:sigma}
Consider a model where the default time $\tau^S$ is defined as the first passage time of a time-homogeneous diffusion martingale $S$ below a given threshold $k<0$, i.e.,
$$
\tau^S:=\inf\left\{t:S_t\le k\right\}\quad\text{where}\quad S_t:=\int_0^t \sigma(s)dB_s,\quad k<0,
$$
and $B$ is a Brownian motion. Then, the survival function of $\tau^S$ is $G^{\tau^S}=G$ if the diffusion coefficient is
\beq\label{eq:sigma}
\sigma(t):=-k\sqrt{\frac{g(t)}{\left(\Phi^{-1}\left(\frac{1+G(t)}{2}\right)\right)^3\phi\left(\Phi^{-1}\left(\frac{1+G(t)}{2}\right)\right)}}.
\eeq
\end{corollary}

\begin{proof} Observe first that $\Sigma^2(t):=\int_0^t\sigma^2(s)ds=\Theta^{G,W}_k(t)$ given in \eqref{eq:ThetaCal}, i.e., that $\sigma^2(t)=\theta^{G,W}_k(t)$ is the corresponding clock rate. On the other hand, $S$ is a local martingale satisfying $\langle S\rangle_t= \Sigma^2(t)= \Theta^{G,W}_k(t)$. Because we know from Theorem \ref{th:clock} that $\Theta^{G,W}_k$ is a clock, $\lim_{t\to\infty} \langle S\rangle_t=+\infty$. From Dambis, Dubins \& Schwarz theorem \cite[Theorem 4.6]{Kara05}, we conclude that the continuous martingale $S$ can be written as $S_t=W_{\langle S \rangle_t}=W_{\Theta^{G,W}_k(t)}$ for some Brownian motion $W$. This proves that
$$
\tau^S=\inf\{t:S_t\leq k\}=\inf\left\{t:W^{\theta^{G,W}_k}_t\leq k\right\}=\tau^{\theta^{G,W}_k}
$$
hence, $G^{\tau^S}(t)=\Pr(\tau^{\theta^{G,W}_k}> t)=\Pr\left(\underline{W}_{\Theta^{G,W}_k(t)}> k\right)=G^{W}_k(\Theta^{G,W}_k(t))=G(t)\;,~~\forall t\ge 0$.
\end{proof}

Observe that the calibration of $G^{\tau^S}$ to $G$ holds automatically and for any threshold $k<0$ both when $S_t=W^\theta_t$ and when $S_t=\int_0^t\sigma(u)dW_u$: it suffices to take the clock $\Theta$ as in \eqref{eq:ThetaCal} in the first setup, and to take the diffusion coefficient $\sigma$ as in \eqref{eq:sigma} in the second setup. 

\begin{remark}\label{rem:ThetaBM:Explosion} In~\cite{OverbeckSchmidt2005}, the authors observe that the clock increases sharply near the origin. They interpret this behavior as a mean to boost the likelihood that the model crosses the barrier quickly.\footnote{A similar observation was made in ~\cite{brigomorinitarenghi} for the piecewise constant volatility parametrization of the AT1P model.} We confirm this observation, and show that this sharp increase is actually a singularity. 
Indeed, expression \eqref{eq:sigma} may lead to $\sigma(0):=\lim_{t\downarrow 0}\sigma(t)=\infty$ depending on the behavior of $G$ around 0. Intuitively, this is because a process $X$ with continuous sample paths, such as $X=W$, cannot cross a barrier $k<X_0$ instantaneously at $t=0$: one needs an exploding clock rate $\theta(t)=\sigma^2(t)$ near 0 in order for an FPT model displaying default predictability to match a curve $G$ that is strictly decreasing from 0. 
Mathematically, for the Brownian latent model,
\[\sigma^2(t)=\theta_k^{G,W}(t) =\frac{g(t)}{g_k^W(\Theta_k^{G,W}(t))} \quad\text{for almost every }t>0. \]
If the target hazard rate has a strictly positive right limit at the origin\footnote{This is typically the case in most practical applications. In the ISDA model, for instance, the $G$ curve exhibits a piecewise constant hazard rate function $h$, where the levels $h_i$, $i=1,\ldots,n$, are all strictly positive.},
\[\lim_{t\downarrow0}h(t)=h(0+)>0, \]
then $g(t)=h(t)G(t)\to h(0+)>0$. On the other hand,
\[\lim_{t\downarrow0}g_k^W(\Theta_k^{G,W}(t))=0,\]
because $\Theta_k^{G,W}(t)\downarrow 0$ and the Brownian first-passage density vanishes at operational time zero. 
Hence the clock rate, and thus $\sigma^2(t)$, may explode at the origin. This behavior may jeopardize the validity of the approach, as this would lead to a firm-value process $S$ with infinite quadratic variation. 

This singularity is, however, always harmless: because $\Theta(t)=\Sigma^2(t)$, any pair of survival curve and latent model for which $\Theta$ is a clock leads to $\Sigma^2(t)<\infty$ for all $t\geq 0$. Squared-integrability of $\sigma$ is thus tightly connected to the existence of a ``clock'' solution to the inversion problem in Lemma~\ref{lem:iclock}. As a consequence,  Assumption $\cA_1$ does not only guarantee the existence and uniqueness of the implied clock, it also ensures the squared-integrability of the volatility function, with no further condition required. The explicit proof of
\[\Sigma^2(t):=\int_0^t\sigma^2(u)\,du<\infty\]
for every $t>0$ and for every regular survival function $G$ in the Brownian case is given in Appendix~\ref{app:sigmaL2} (Theorem~\ref{th:sigma}), as a sanity check.
\end{remark}

\section{Extensions}
In Section \ref{sec:FPT-BM}, we have shown that it is easy to design a structural model displaying the perfect calibration feature: it suffices to define $\tau$ as the FPT of a time-changed version $X^\theta$ of a latent process $X=W$ below a threshold $k<0$, where the time-change function $\Theta$ is set to the implied clock, $\Theta^{G,W}_k$. This can be equivalently defined as the first passage time of the stochastic integral $S_t:=\int_0^t\sigma(u)dW_u$ below $k$ provided that the diffusion coefficient satisfies $\Sigma^2=\Theta^{G,W}_k$. Despite its ability to reproduce any curve $G$ complying with Assumption \ref{ass:A1}, this setup is restrictive in the sense that it features only one free parameter, namely, the threshold $k$. Although the latter can be chosen arbitrarily \--- for any $k<0$, there exists a clock $\Theta$ offering a perfect calibration to $G$ \--- it matters in the sense that it controls the variance of the corresponding firm-value process: from \eqref{eq:sigma}, it is clear that $\var(X_t)\propto k^2$. Therefore, it is useful to examine more general approaches, where the time change is applied to more sophisticated latent processes $X$ than just a standard Brownian motion.\medskip

Interestingly, Lemma~\ref{lem:iclock} suggests that this approach can be successfully applied to a much broader class of processes $X$: it suffices that $G^{X}_k$ is a regular survival function. Therefore, we consider several extensions below. The first extension is to apply the time-change not to a Brownian motion $W$ but to a drifted (and possibly scaled) Brownian motion $\tilde{W}$, i.e., consider $X=\tilde{W}$ as latent process, and then take $S=X^\theta$ as before. The second extension deals with the case where the firm-value process $S=X^\theta$ is the time-changed version of a latent process $X$ that is a mean-reverting GBM, possibly with a random barrier.

\subsection{Drifted Brownian motion}\label{sec:BMdrift}

The distribution of the running infimum $\underline{\tilde{W}}$ of a Brownian motion with drift $\tilde{W}_t=\mu t+W_t$ below a given $k<0$ is well-known, too (see, e.g., \cite[formula 1.2.4 and 1.2.5]{Boro02}):

\beqn
G^{\tilde{W}}_k(t):=\Pr\left(\inf_{0\leq s\leq t}\mu s+W_s> k\right)=\Phi\left(\frac{\mu t-k}{\sqrt{t}}\right)-e^{2\mu k}\Phi\left(\frac{\mu t+k}{\sqrt{t}}\right)\;.
\label{eq:Mtidle}
\eeqn

Unless $\mu=0$, corresponding to the case $\tilde{W}=W$ discussed in Section~\ref{sec:FPT-BM}, this distribution features a sum of standard Normal distributions and, therefore, its inverse does not admit a closed-form expression. In this case, our method still works, but we lack the analytical expression of the implied clock $\Theta^{G,\tilde{W}}_k$. Nevertheless, the latter can be easily obtained by numerical inversion. This can be done in a very efficient way, computationally speaking.\footnote{Observe that if $\mu\leq 0$ the probability to never hit the barrier is zero, $\lim_{t\to\infty} G^{\tilde{W}}_k(t)=0$, such that $G^{\tilde{W}}_k$ is regular. If $\mu>0$, however, there is a positive probability to never hit the barrier because $\lim_{t\to\infty} G^{\tilde{W}}_k(t)=1-e^{2\mu k}$, in which case $G^{\tilde{W}}_k$ is not regular. This is why, in credit risk applications, one usually take $\mu\leq 0$.}\medskip

Let us now model the default time $\tilde{\tau}^\theta$ as the FPT of a time-changed Brownian motion with constant drift $S=X^\theta$ where $X=\tilde{W}$ below $k<0$. The firm-value process becomes
$$
S_t=\tilde{W}^\theta_t=\mu\Theta(t)+W^\theta_t=\mu(t)+\int_0^t\sigma(s)dB_s,
$$
where $\mu(t):=\mu\Theta(t)$ and $\sigma^2(t)=\Theta'(t)$. It is clear that $G^{\tilde{\tau}^\theta}=G$  provided that $\Theta$ is the implied clock $\Theta^{G,\tilde{W}}_k=I^{\tilde{W}}_k(G(t))$, where $I^{\tilde{W}}_k$ is the inverse of $G^{\tilde{W}}_k$ in \eqref{eq:Mtidle}. In contrast with $W$, the case $\tilde{W}$ now features two parameters: $\mu$ and $k$, which are not redundant.

\begin{remark}\label{rem:scaledBM}
The case of a scaled drifted Brownian motion brings no extra benefit  because the FPT of $\hat{W}_t:= \mu t+ \gamma W_t$ below $k$ obviously coincides with the FPT of $\mu t/|\gamma|+W_t$ below $k/|\gamma|$ which can be recovered from \eqref{eq:Mtidle}.
\end{remark}

\subsection{The AT1P model}\label{sec:AT1P}

A very broad class of FPT models can be calibrated using our time change approach: it suffices to consider the corresponding firm-value process $S$ as our latent process $X$, and define the new firm-value process as $X^\theta$ where $\Theta$ is the clock implied by $G,X$. Let us illustrate this on the Brigo-Morini-Tarenghi (BMT) model, also known as the AT1P model~\cite{brigomorinitarenghi}, which considers a mean-reverting GBM firm-value process $S$ and a barrier $k(t)$ which is time-dependent. Precisely, the default time is modeled as the FPT of $S^{\tiny\text{AT1P}}_t-k(t)$ in the negative territory:\footnote{We use $K_0$ instead of $H$ here to avoid confusion with the integrated hazard rate function.}
\beqn
S^{\tiny\text{AT1P}}_t&=&S_0\exp\left\{\int_0^t(r_u-p_u)-\frac{1}{2}\sigma^2(u)du+\int_0^t\sigma(u) dW_u\right\}\nonumber\\
k(t)&=&K_0\exp\left\{\int_0^t r_u - p_u - B\sigma^2(u) du\right\}
\eeqn
where $r$ is the risk-free rate process, $p$ the payout ratio process, $\sigma$ the (deterministic) instantaneous volatility and $K_0,B$ are constants satisfying $0<K_0<V_0$.\footnote{To remain in the regular setting, one needs to impose a negative drift $\mu\le 0$, i.e., $B\le 1/2$. Otherwise there is a non-zero probability of never hitting the barrier, i.e., $\lim_{t\to\infty}G^{\text{AT1P}}(t)>0$. Also, the authors of AT1P considers a piecewise constant volatility parametrization when calibrating the model to a CDS term-structure. We consider that this feature is not inherent to the definition of AT1P, but is just one possible parametrization of the AT1P model.} This model leads to a tractable expression for the survival probability curve:
$$
G^{\tiny\text{AT1P}}(t):=\Pr(\tau^{\tiny\text{AT1P}}>t)\quad\text{where}\quad \tau^{\tiny\text{AT1P}}:=\inf\{t:S^{\tiny\text{AT1P}}_t-k(t)\leq 0\},
$$
which is given by
\beq
G^{\tiny\text{AT1P}}(t)=\Phi\left(\frac{\mu\Sigma^2(t)-\log\frac{K_0}{S_0}}{\Sigma(t)}\right)-\left(\frac{K_0}{S_0}\right)^{2\mu}\Phi\left(\frac{\mu\Sigma^2(t)+\log\frac{K_0}{S_0}}{\Sigma(t)}\right)\label{eq:G-AT1P}
\eeq
where $\mu:=B-1/2$ and  $\Sigma^2(t)=\int_0^t\sigma^2(u)du$, as before.\footnote{This expression can be easily recovered thanks to our time-change trick; see Appendix~\ref{app:AT1PasTC}.}\medskip

Clearly, the processes $r,p$ play no role in the survival probability. For a given pair of parameters $(K_0,B)$, the authors propose to choose the volatility function so as to best fit the term-structure of CDS. Precisely, the volatility function is made piecewise constant between the available CDS tenors $T_1,\ldots,T_n$, i.e., $\sigma(t)=\sigma_i$ for all $T_{i-1}<t\leq T_i$, where the levels $\sigma_1,\ldots,\sigma_n$ are computed iteratively to match the market quotes. We refer to Section~\ref{app:CalAT1P} for a discussion about this specific parametrization of AT1P.\medskip

Our extension consists in taking the AT1P firm-value process with constant colatility $\sigma(t)=\sigma$ as latent process\footnote{That is, replace $\Sigma^2(t)$ by $\sigma^2 t$ in \eqref{eq:G-AT1P}.}, $X=S^{\tiny\text{AT1P}}$, and time-changing both $X$ and $k$ using a clock $\Theta$ so as to fit $G$.\footnote{One could start from a time-varying $\sigma$ in AT1P but its effect will be completely absorbed in the time change, such that there is nothing to gain in doing so.} This procedure is possible because \eqref{eq:G-AT1P} is a regular survival function, and has three advantages. First, the calibration step, which in the AT1P model consisted in building the piecewise constant function $\sigma$, now simply amounts to compute the implied clock of an AT1P model with a fixed $\sigma$, which can be easily recovered by simple numerical inversion of $G^{\tiny\text{AT1P}}$ in \eqref{eq:G-AT1P}: $\Theta^{G,\tiny\text{AT1P}}(t)=I^{\tiny\text{AT1P}}(G(t))$. This is a substantial benefit when the model needs multiple recalibrations, e.g., in case of joint calibration to $G$ and option prices, for instance. Second, our approach allows to fit the full functional $G$, not only a finite set of spreads or probabilities. This is important for instance if one needs the model to be calibrated to the full ISDA model, assuming piecewise constant hazard rate functions, as this would require the point-by-point identification of the $\sigma$ function. Finally, no mathematical evidence was provided to show that there always exists a function $\sigma$ piecewise constant between the tenors, allowing AT1P to match any given term structure of CDS spreads.\footnote{However, the authors show from the Lehman case that the model can reproduce backwardation markets, suggesting that the AT1P model with piecewise constant volatility parametrization is quite flexible.} In contrast, it is known from Lemma~\ref{lem:iclock} that the implied clock $\Theta^{G,\tiny\text{AT1P}}$ exists and is unique.\medskip

 Interestingly, the AT1P model itself (i.e., not the time-changed AT1P with constant volatility discussed above, but the original AT1P model with the time-dependent $\sigma$) displays very strong connections with our time-change approach. In fact, the AT1P model is nothing else than the FPT of a drifted Brownian motion $\tilde{W}$ with drift $\mu=B-1/2$, time-changed using $\Theta:t\mapsto \Theta(t):=\Sigma^2(t)$ \--- which is indeed a clock in the sense of Definition \ref{def:clock} \--- below a barrier $k=\log(K_0/S_0)<0$. This is relatively obvious when comparing \eqref{eq:Mtidle} with \eqref{eq:G-AT1P}: replacing $t$ by $\Sigma^2(t)$ in the former yields the latter. In other words, $G^{\tiny\text{AT1P}}(t)=G^{\tilde{W}}_{k}(\Sigma^2(t))$ where $k=\log(K_0/S_0)$ and the drift of $\tilde{W}$ is $\mu=B-1/2$; see Appendix \ref{app:AT1PasTC}.\medskip
 
 This relationship between AT1P and the time-changed Brownian motion with drift discussed in Section~\ref{sec:BMdrift} allows us to make three points. First, the original calibration of AT1P poposed in~\cite{brigomorinitarenghi,Brigo2004}, which consisted in finding a volatility function $\sigma:t\to\sigma(t)$, essentially amounts to determine the clock implied by the FPT of a drifted Brownian motion, $\Theta^{G,\tilde{W}}_{\log(K_0/S_0)}$.\footnote{This is not exactly the same because the authors propose to determine the levels $\sigma_i$ to match a set of spreads, but one could have equivalently decided to set the levels to match instead the SP curve of some standard bootstrapping model such as ISDA at the tenors $T_i$.} Second, the interpretation of the deterministic volatility function featured in the original AT1P model is exactly to act as a time change, explaining why it offers the flexibility required for the calibration exercise. Third, this connection illuminates the role of the diffusion coefficient and explains how one can calibrate AT1P not just to match a set of quotes, but to fit a continuous SP curve: it suffices to impose the volatility function $\sigma(\cdot)$ in AT1P such that $\Sigma^2(t)=\Theta^{G,\tilde{W}}_{\log(K_0/S_0)}(t)$.\footnote{This interpretation actually proves that there always exists a set $\sigma_1,\ldots,\sigma_n$ allowing the AT1P model to fit probabilities at the tenors $T_1,\ldots,T_n$ provided that $G$ satisfies Assumption $\ref{ass:A1}$; see Appendix~\ref{app:CalAT1P}. This calibration is different from but very similar to that considered by Brigo et al., which is based on the CDS par spreads.}\medskip

The authors of \cite{brigomorinitarenghi} noticed that, in AT1P, the short-term volatility $\sigma(t)$ for $t$ close to 0, i.e., $\sigma_1$, is extremely large when the model assumes a piecewise constant volatility function calibrated to actual market quotes. Our analogy helps understand the reason. Indeed, recall from Remark~\ref{rem:ThetaBM:Explosion} that $\theta^{G,W}_k(t)$, i.e., the coefficient $\sigma(t)$ in \eqref{eq:sigma}, explodes as $t\downarrow 0$. The same happens for $\theta^{G,\tilde{W}}_k(t)$. The intuition is related to the default predictability inherent to structural models, and was discussed in Remark~\ref{rem:ThetaBM:Explosion}. 

More precisely, let $\Theta=\Theta_k^{G,\widetilde W}$ be the clock implied by the drifted-Brownian latent model. Its rate is
\[\theta(t)=\frac{g(t)}{g_k^{\widetilde W}(\Theta(t))}\quad\text{for almost every }t>0.\]
If the target hazard rate admits a strictly positive right limit at the origin, i.e.,
\[\lim_{t\downarrow0}h(t)=h(0+)>0,\]
then $g(t)=h(t)G(t)\to h(0+)>0$. On the other hand,
\[g_k^{\widetilde W}(\Theta(t))\to0,\qquad t\downarrow0,\]
because $\Theta(t)\downarrow 0$ and the first-passage density of the
drifted Brownian latent process vanishes at operational time zero.
Consequently, the implied clock rate, which is the implied variance rate
$\sigma^2(t)$, diverges as $t\downarrow 0$.\medskip

This feature is the familiar short-end limitation of continuous-path structural models with a deterministic positive initial distance to default: the default probability is smaller than first order as $t\downarrow 0$, so the instantaneous credit spread vanishes. It can be mitigated by an exploding short-end variance rate, jumps in the firm-value process, or a randomized barrier with mass sufficiently close to the initial firm value; see, for example, the discussion and the SBTV construction in \cite{Brigo13,brigomorinitarenghi,Brigo2005}. In \cite{Brigo13} a proof is given showing that, for the Merton model, the short term intensity of credit spreads tends to zero when the maturity tends to zero. The proof can be extended to Black Cox. 

Nevertheless, even if the chosen latent model suffers from the vanishing short-term spread problem, it will be resolved thanks to the implied clock. Indeed, the CDS term structure only depends on the survival probability curve $G$. Therefore, the perfect fit of the model implied curve to the market survival curve guarantees that they will display the same short-term spreads behavior, too. In our setting, this is possible because the implied clock rate can explode near the origin without breaking square-integrability, as explained in Remark~\ref{rem:ThetaBM:Explosion}.\medskip

The numerical examples in Figure~\ref{fig:AT1P} illustrate this effect for the original AT1P specification and for its time-changed representation. The top panels correspond to $G$ with $h(t)=10\%$, with $k=-1$ and $k=-5$, respectively. The bottom panels correspond to the target curve $G$ built from a piecewise constant hazard function with $k=-1$.

\subsection{Other latent models}

The two latent models considered above can be recast as FPT of Brownian motions with linear drift below a constant barrier. However, this is not a restriction to our approach. For instance, one could consider jump-diffusion dynamics for $X$ and/or random barriers for $k$. The only constraint is that the model associated with $(X,k)$ admits a regular survival curve.\medskip

For instance, consider the case where $k(0)$ is random. An example of such models is the scenario-based threshold volatility (SBTV) of Brigo et al.~~\cite{brigomorinitarenghi}, where $k(0)=K_0$ is a discrete random variable with two states. It is easy to see that our approach can deal with such cases, under mild conditions. For instance, replacing $K_0$ by $S_0e^{-\tilde{h}}$ leads to a survival curve
\[
G^{X}(t)=\E\left[\Phi\left(\frac{\mu\gamma^2 t+{\tilde{h}})}{\gamma\sqrt{t}}\right)-e^{-2\mu {\tilde{h}}}\Phi\left(\frac{\mu\gamma^2 t-{\tilde{h}}}{\gamma\sqrt{t}}\right)\right]\;.
\]

This is a regular survival curve whenever $\mathbb P(K_0<S_0)=1$ (i.e., $\tilde h>0$ a.s.) and $\mu\le 0$; the latter condition, as in the single-scenario case, ensures that the survival probability vanishes as $t\to\infty$ for every realization of the barrier, and hence after mixing over scenarios.

\section{Numerical Examples}\label{sec:simu}
\subsection{Time-changed Brownian Motion}
We set $k=-1$ and $G(t)=e^{-ht}$ with $h=10\%$. We first check the distribution of the running infimum of a Brownian motion as follows. We generate $N=1,000$ trajectories of a Brownian motion $W$ over $[0,T]$ with time step $\delta t=0.01$ and $T=n\delta t=20$. We then track the running infimum $M$ by computing, for each $t_i:=i\delta t$, $i\in\{0,1,\ldots,n\}$, $M[,i]=\min(M[,i-1],W[,i])$. We then plot on Figure \ref{fig:BM1} below the theoretical distribution function $\Pr(M_{t_i}>k)$ in \eqref{eq:RMBM} (black, solid) with the empirical version $\hat{\Pr}(M_{t_i}>k):=\frac{1}{N}\sum_{j=1}^N \ind_{\{M[j,i]>k\}}$ (red, dots). The two curves agree, as one can check on Fig. \ref{fig:BM1}. We then do the same for the time-changed process
\[X_t=\int_0^t\sigma(s)\,dW_s=W_{\Theta(t)},\]
where $\Theta=\Theta_k^{G,W}$ is the implied clock. At the monitoring dates,
we simulate $X$ by its exact Gaussian-increment scheme
\[X[j,i]=X[j,i-1]+\sqrt{\Theta(t_i)-\Theta(t_{i-1})}\,Z[j,i-1],\]
where the $Z[j,i]$ are i.i.d.\ standard Normal random variables. The
running infimum is monitored only on this time grid and hence provides a
discretely monitored approximation to the continuous-time running infimum.
Again, the empirical distribution of the discretely monitored running
minimum of $X$ agrees closely with $G$; see Fig.~\ref{fig:TC-BM1}.

The implied clock $\Theta^{G,W}_k$ is drawn on Fig. \ref{fig:Clock-BM1}. The panels on the middle row, Fig. \ref{fig:BM}, \ref{fig:TC-BM} and \ref{fig:Clock-BM}, show similar results when $k=-5$. Finally, we the panels on the last row correspond to the case where $G$ is built from the piecewise constant hazard rate function given in Table 2 in \cite{brigomorinitarenghi}.

\begin{figure}
\centering
\subfigure[Brownian motion]{\includegraphics[width=0.32\columnwidth]{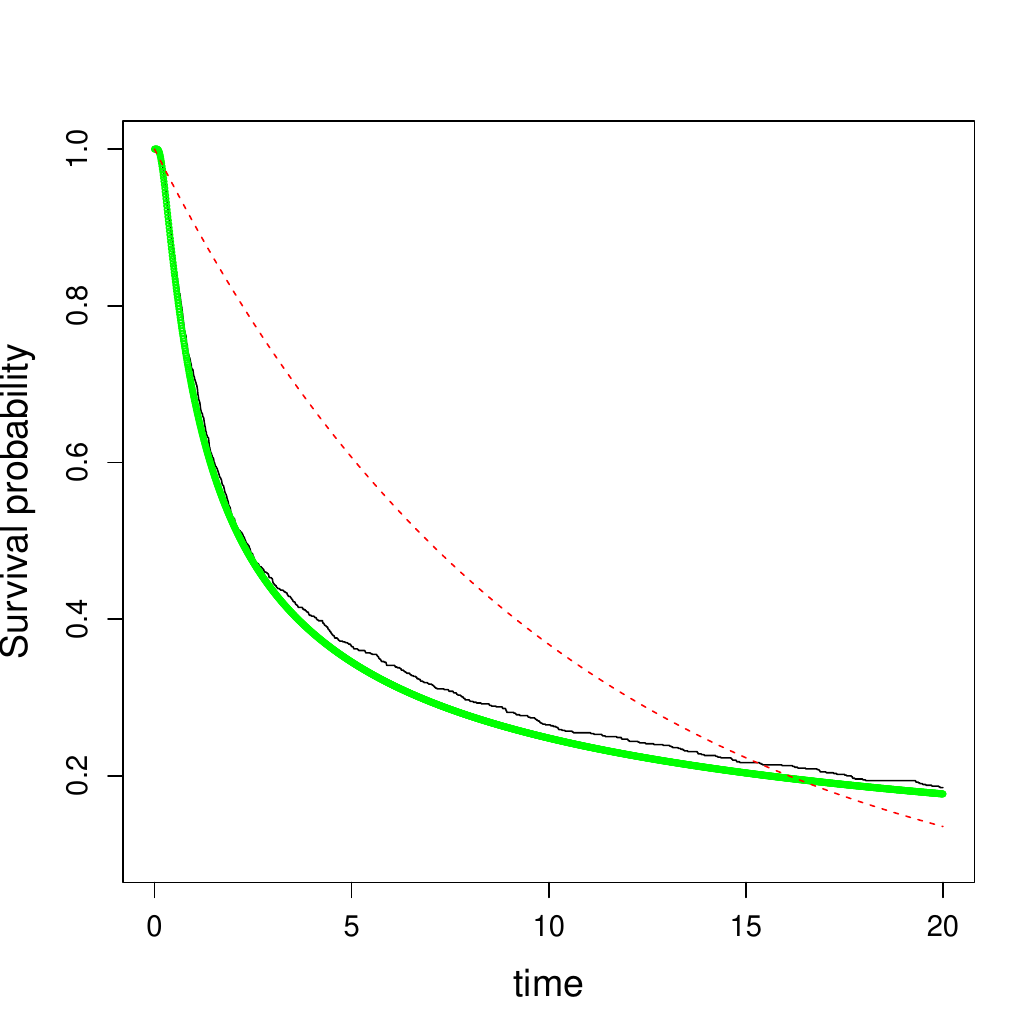}\label{fig:BM1}}
\subfigure[TC Brownian Motion]{\includegraphics[width=0.32\columnwidth]{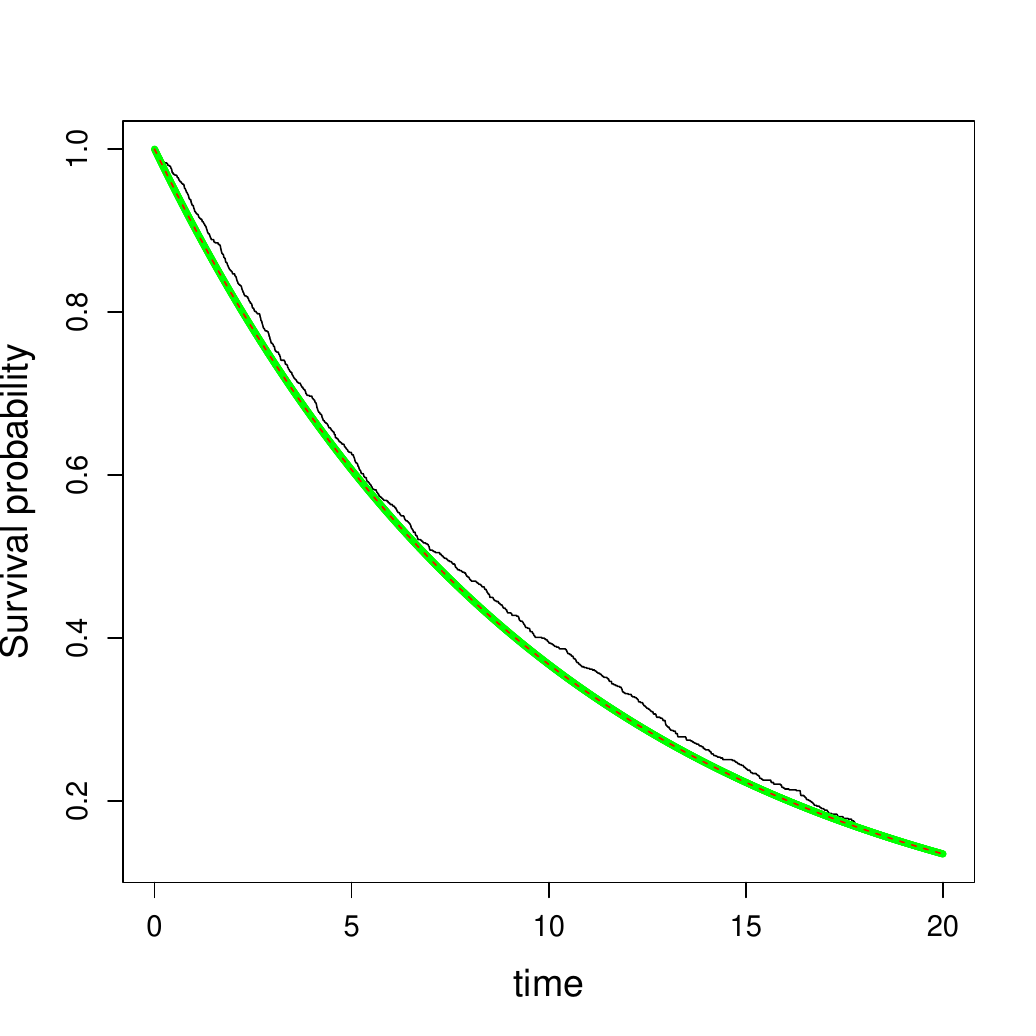}\label{fig:TC-BM1}}
\subfigure[Implied clock]{\includegraphics[width=0.32\columnwidth]{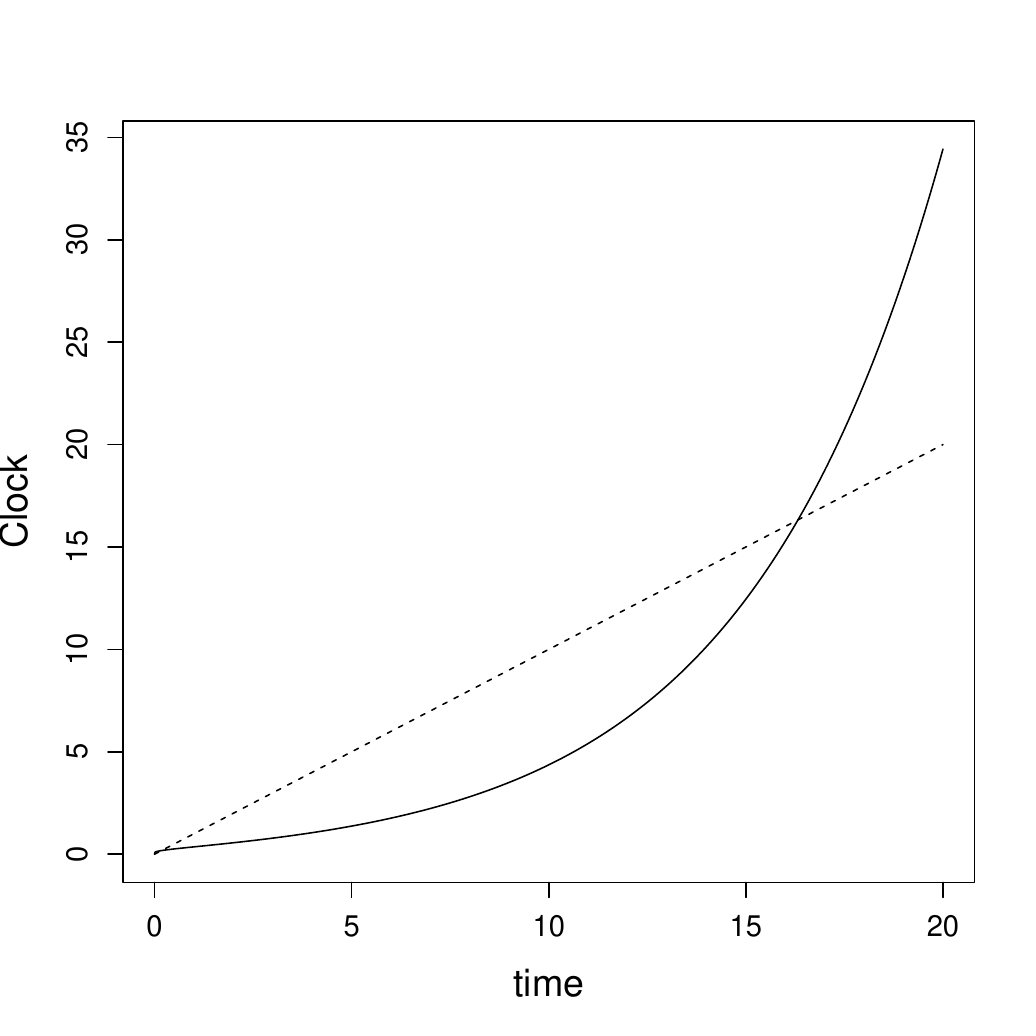}\label{fig:Clock-BM1}}\\
\subfigure[Brownian motion]{\includegraphics[width=0.32\columnwidth]{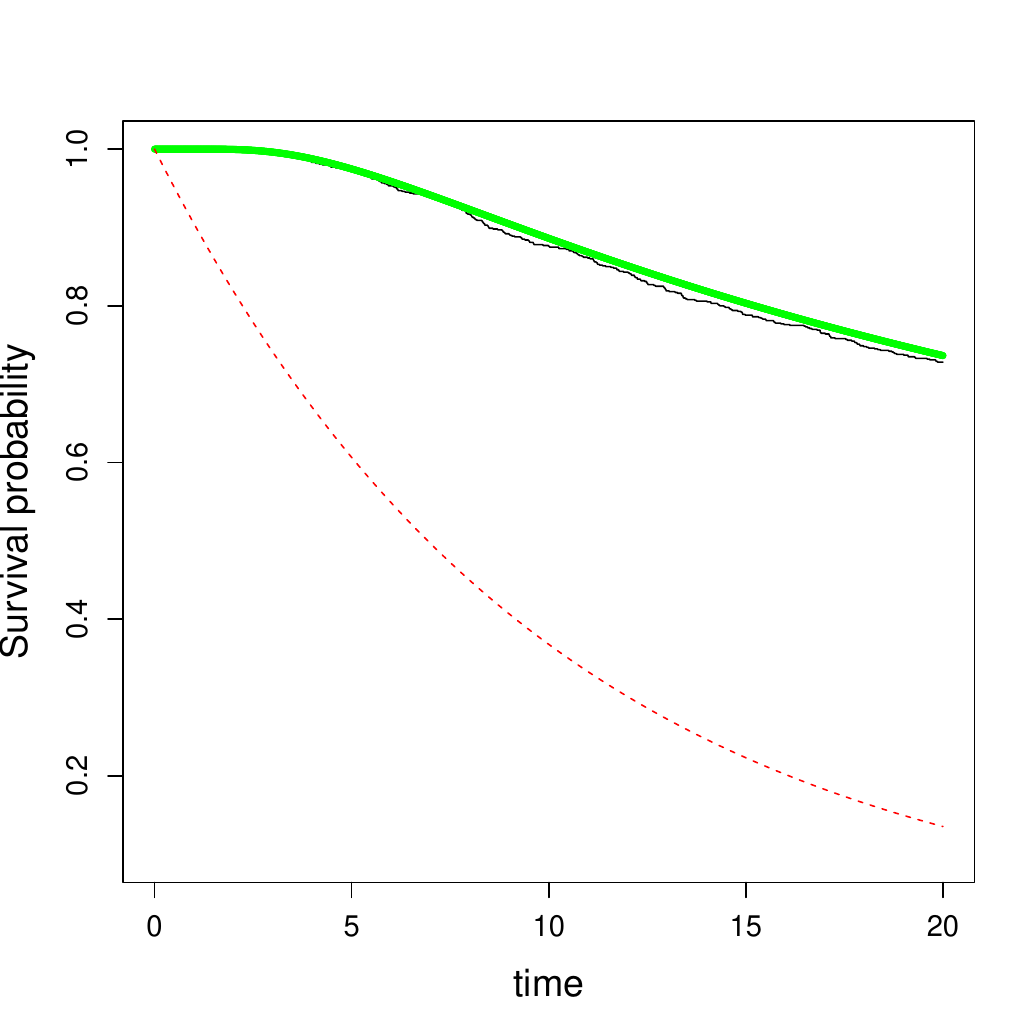}\label{fig:BM}}
\subfigure[TC Brownian motion]{\includegraphics[width=0.32\columnwidth]{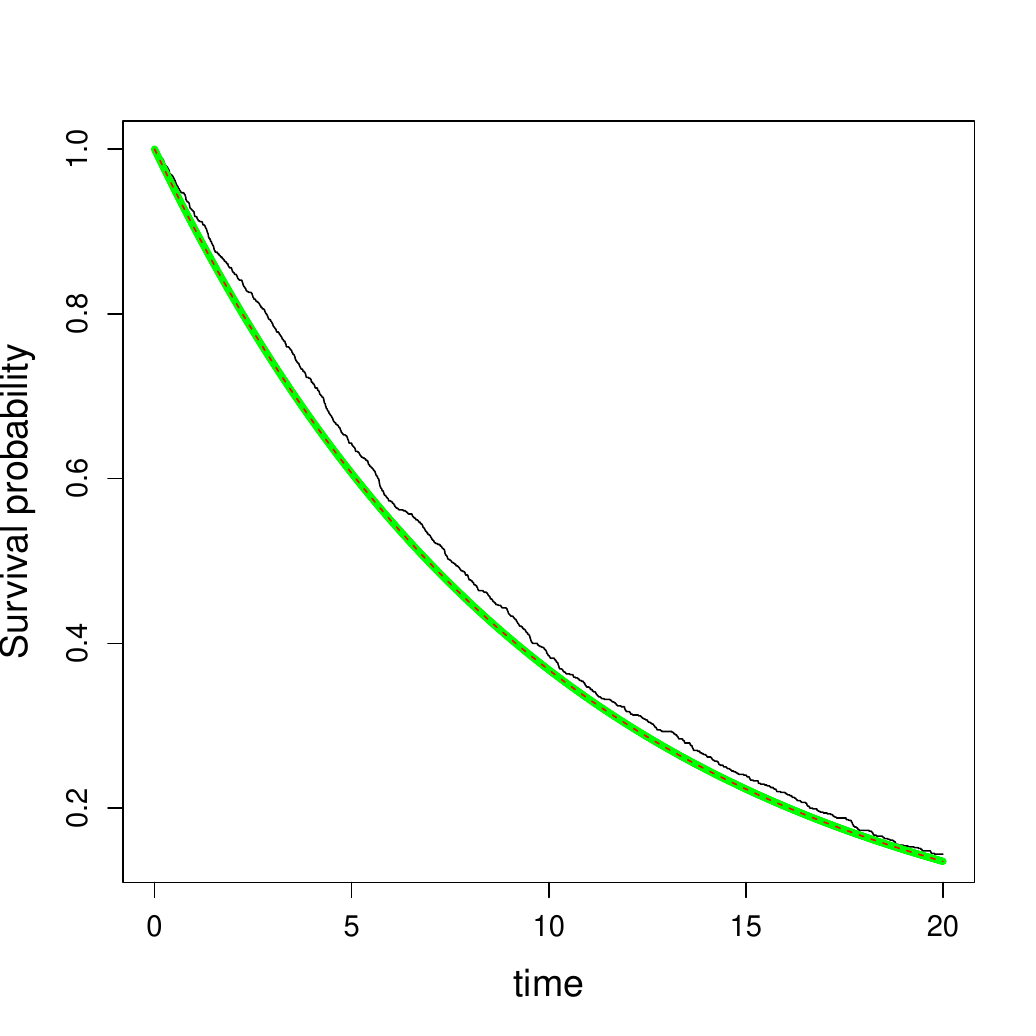}\label{fig:TC-BM}}
\subfigure[Implied clock]{\includegraphics[width=0.32\columnwidth]{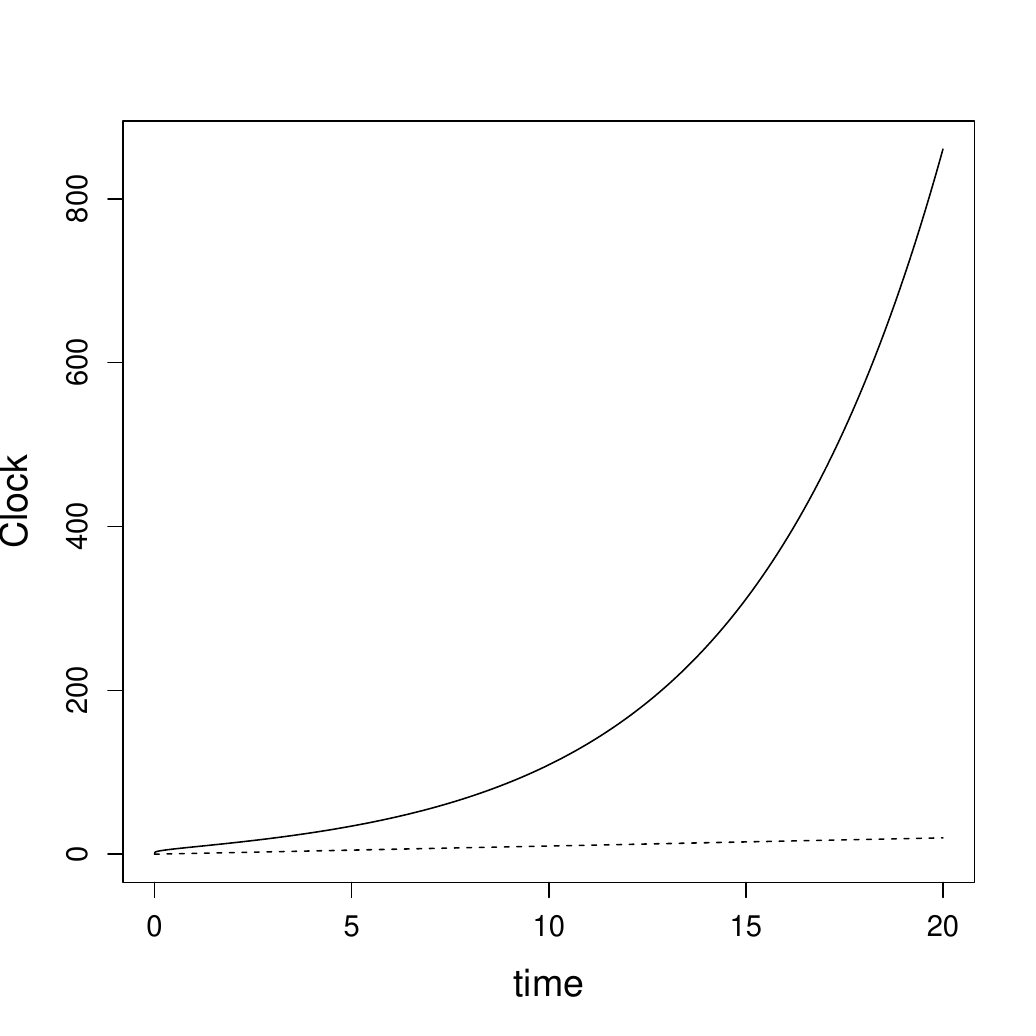}\label{fig:Clock-BM}}\\
\subfigure[Brownian motion]{\includegraphics[width=0.32\columnwidth]{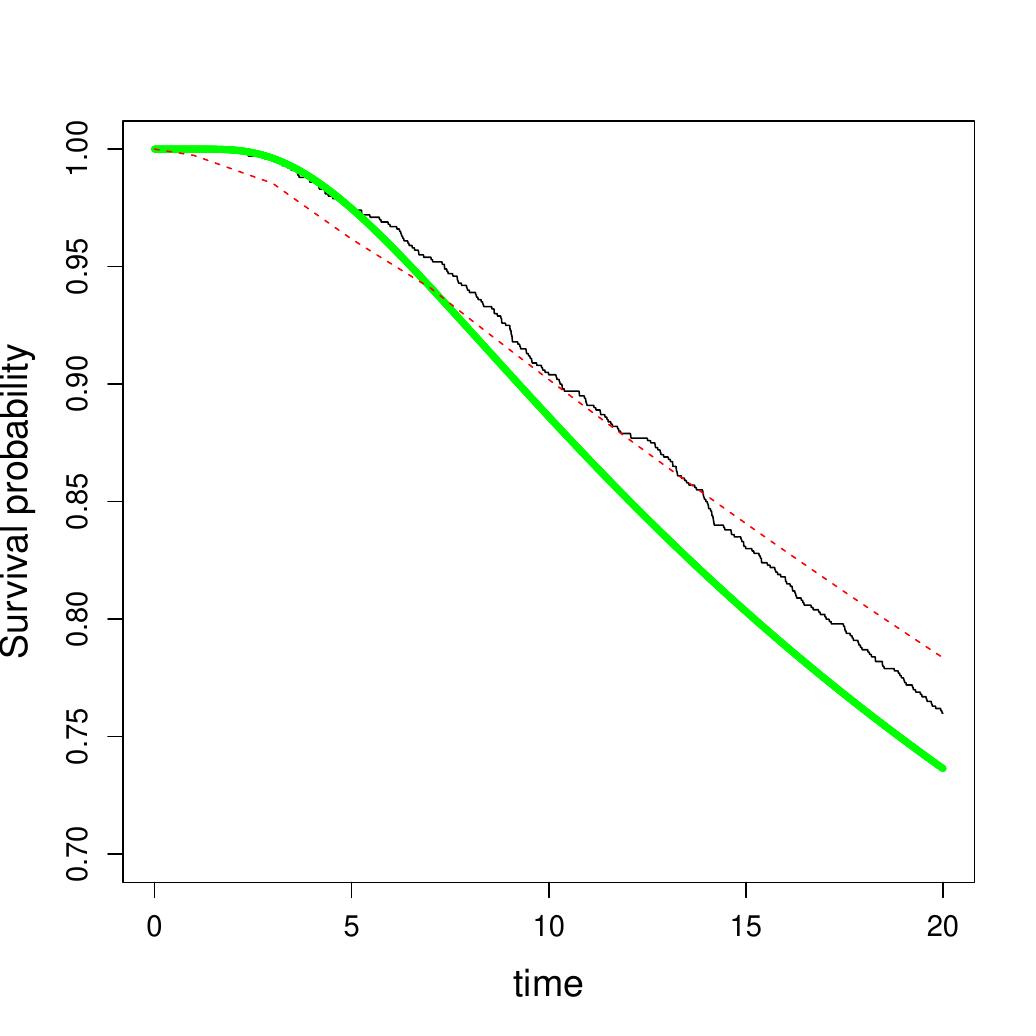}\label{fig:BM-pwc}}
\subfigure[TC Brownian motion]{\includegraphics[width=0.32\columnwidth]{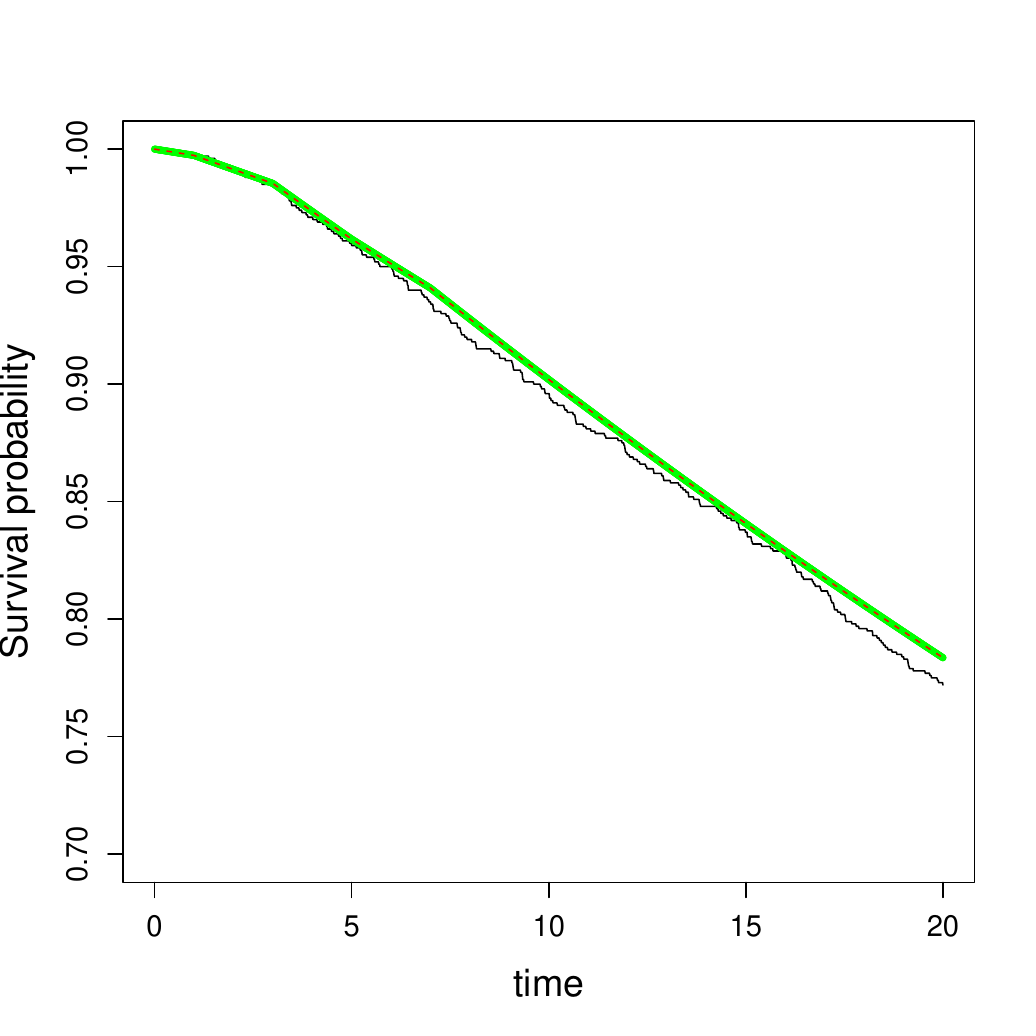}\label{fig:TC-BM-pwc}}
\subfigure[Implied clock]{\includegraphics[width=0.32\columnwidth]{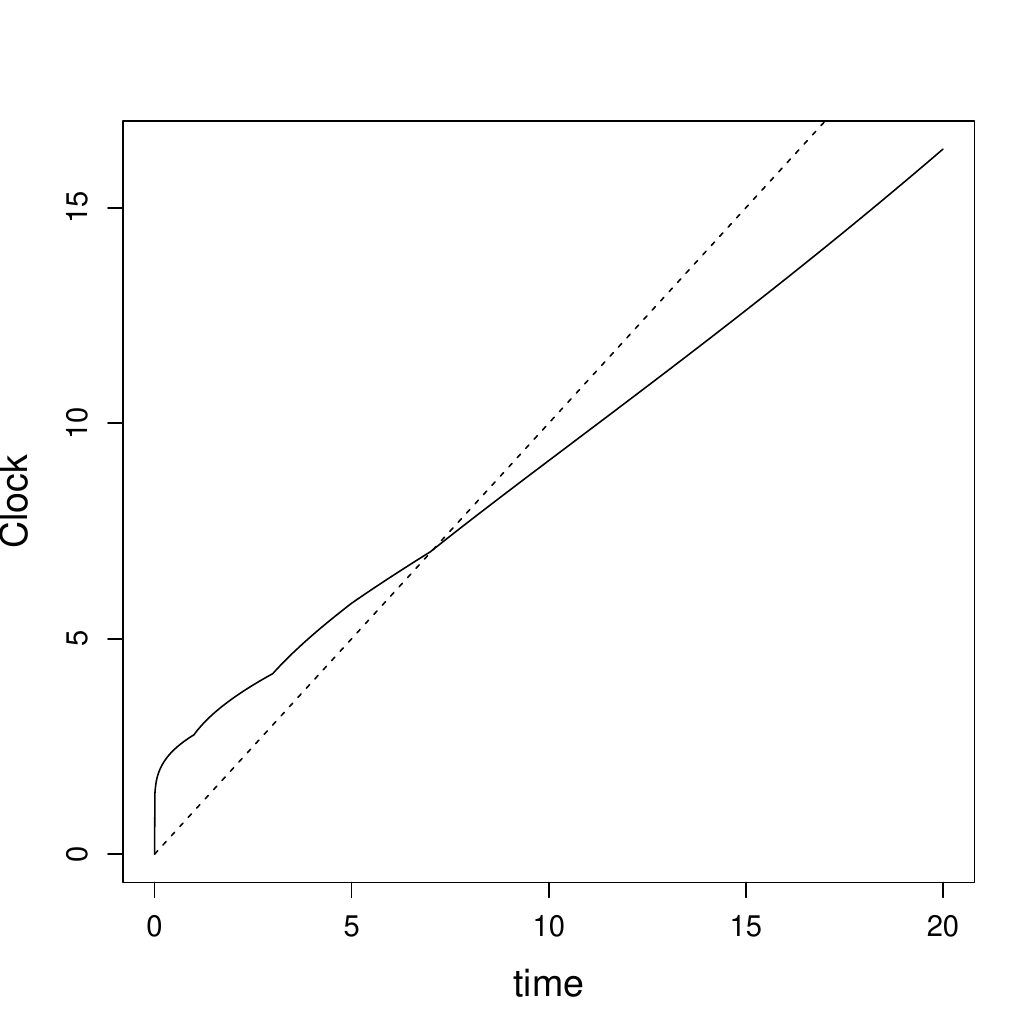}\label{fig:Clock-BM-pwc}}
\caption{Time-changed Brownian motion model and its perfect fit feature. Left: empirical (black) and theoretical (green) survival distributions of the running infimum $\underline{W}$ of a Brownian motion $W$ and the target curve $G(t)$ (red). Center: the empirical (black) and theoretical (green) survival distribution function of the default time in the time-changed Brownian motion with clock $\Theta^{G,W}_k$ given by \eqref{eq:ThetaCal} agree with $G$ (red). Right, the implied clock $\Theta^{G,W}_k$ is evolving around $t$: it varies around the calendar time to compensate for the discrepancies. Top: $G(t)=\exp(-0.1 t)$ and $k=-1$.
Middle: $G(t)=\exp(-0.1 t)$ and $k=-5$. Bottom: $G$ built from the piecewise constant hazard rates in Table 2 of \cite{brigomorinitarenghi} and $k=-5$.}\label{fig:GlobBM1}
\end{figure}

\subsection{AT1P}

We take the model in Section~\ref{sec:AT1P} of with parameters $H=0.4, S_0=1,\sigma(u)=\sigma=0.15$ and $B=0$. We note the corresponding survival probability curve $G^{AT1P}(t;0,0.15,1,0.4)$; it corresponds to the green line on Fig.~\ref{fig:AT1P} This theoretical expression is in line with the first passage time empirical distribution found by Monte Carlo. We then ask ourselves how we can fit this model to the curve $G$ built using the piecewise constant hazard rate function given in Table 2 of \cite{brigomorinitarenghi}. From the previous discussion, it is enough to compute the implied clock $\Theta(t):=\Theta^{G,\tilde{W}}_{\log(K_0/S_0)}=H^{AT1P}(G(t);0,0.15,1,0.4)$. We compare on Fig. ~\ref{fig:TC-AT1P} the $G$ curve (red) with the empirical distribution found by looking at the first passage time of $S^\theta_t=S_{\Theta(t)}$ to the barrier $k^\theta(t)=k(\Theta(t))$ (black). The simulation of $S^\theta$ is done as follows: $S^\theta[,i]=S_0e^{-\Theta(t_i)²/2\sigma^2+X[,i]}$ where $X[,i]=X[,i-1]+\sigma\sqrt{\frac{\Theta(t_i)-\Theta(t_{i-1})}{\delta t}}(W[,i]-W[,i-1])$ and $X[,0]=0$. Then, we track the running infimum $M$ of the difference $V^\theta[,i]-k^\theta(t_i)$, i.e. $M[,i]=\min(M[,i-1],V^\theta[,i]-k^\theta(t_i))$, and eventually estimate the time-$t_i$ survival probability empirically as $\frac{1}{N}\sum_{j=1}^N \ind_{\{S[j,i]\geq 0\}}$.\footnote{In the Brownian case with fixed barrier, one could easily correct for the possibility of $W$ to hit $k$ between two grid points using a Brownian Bridge; see e.g.~\cite{OverbeckSchmidt2005} or~\cite{Vrins26}. However, we do not apply this enhancement here as our main purpose here is to illustrate the theoretical results.} The two curves agree showing that the first passage time of the time-changed firm-value process $S^\theta$ to the barrier $k^\theta_t$ agrees with $G$: the calibration is effective. These curves also match the theoretical survival probability curve of the time-changed AT1P with fixed volatility (green). The implied clock $\Theta=\Theta^{G,\tilde{W}}_{\log(K_0/S_0)}$ is shown on Figure~\ref{fig:Clock-AT1P}. The bottom panels of Fig.~\ref{fig:AT1P} display the same figures for $B=1/2$.

\begin{figure}
\centering
\subfigure[AT1P model]{\includegraphics[width=0.32\columnwidth]{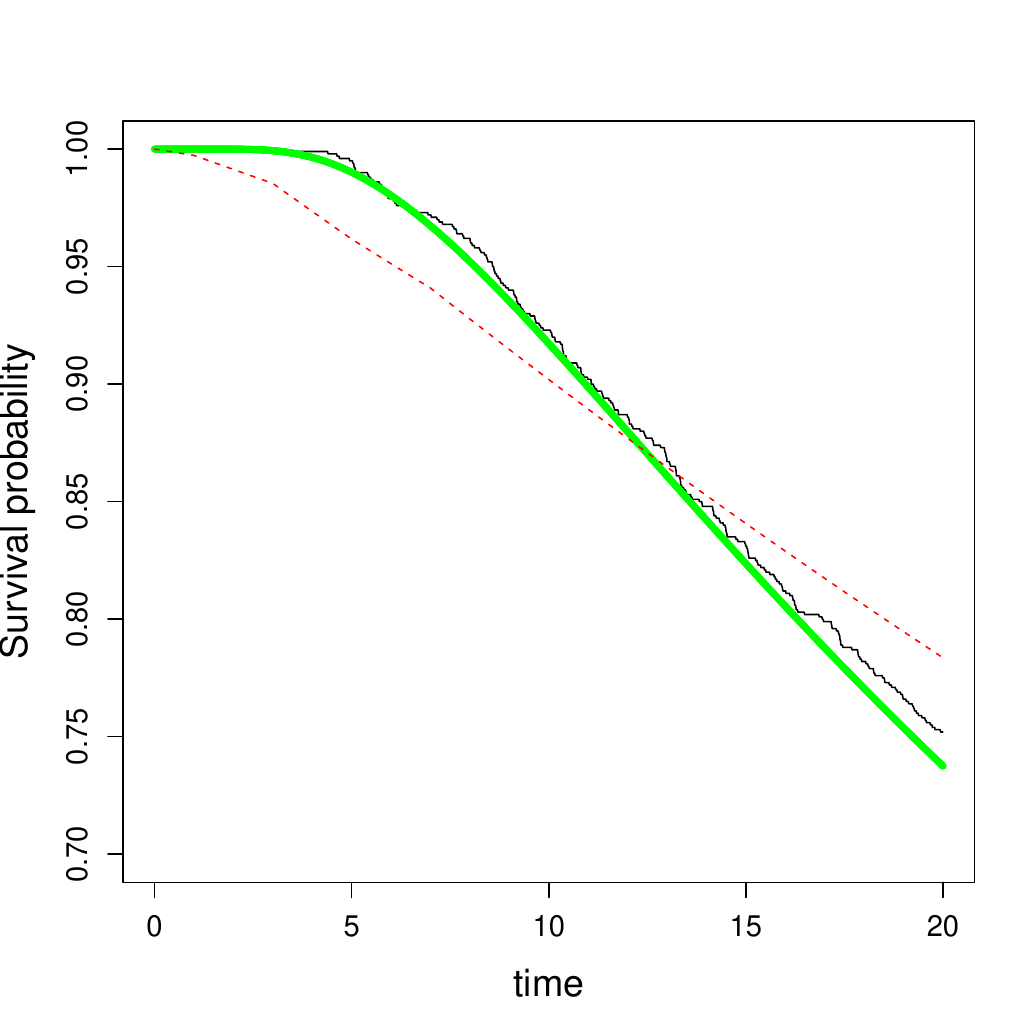}\label{fig:SP-AT1P}}
\subfigure[TC AT1P model]{\includegraphics[width=0.32\columnwidth]{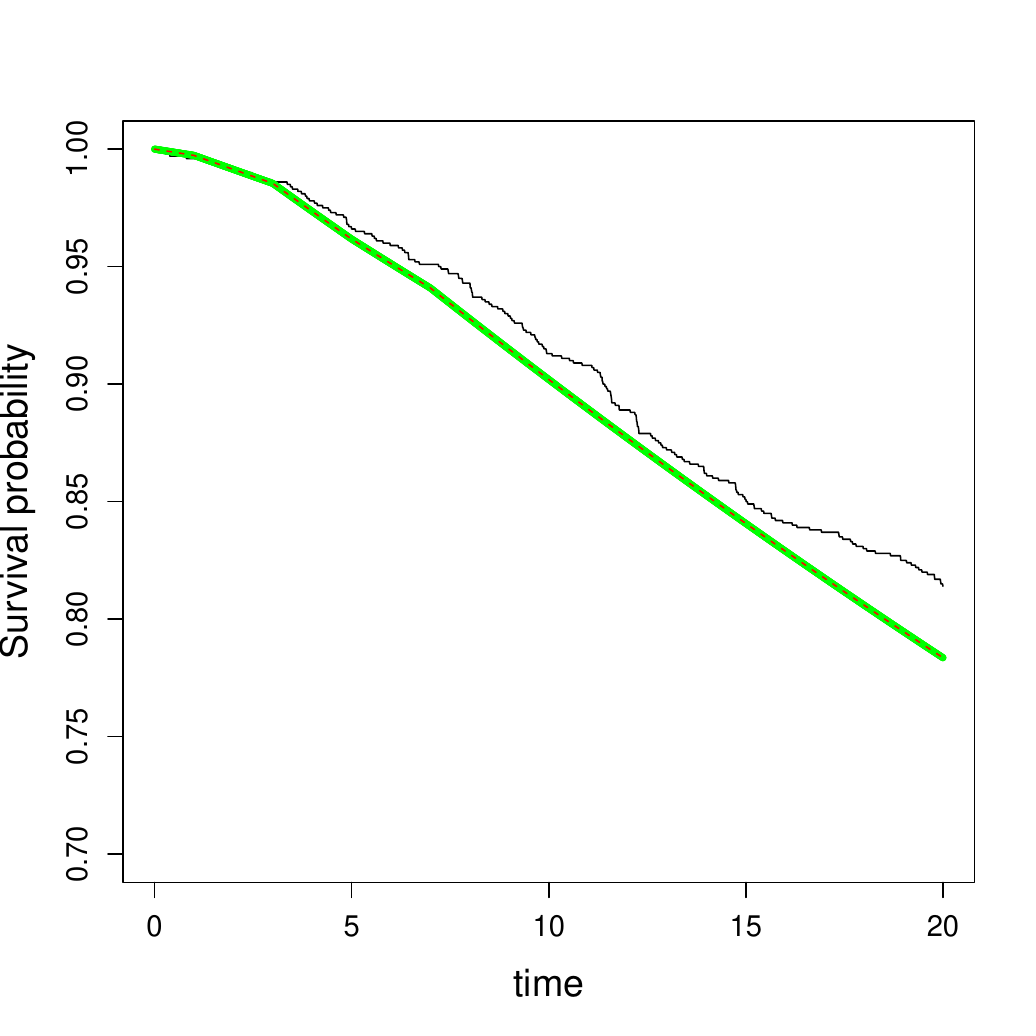}\label{fig:TC-AT1P}}
\subfigure[Implied clock]{\includegraphics[width=0.32\columnwidth]{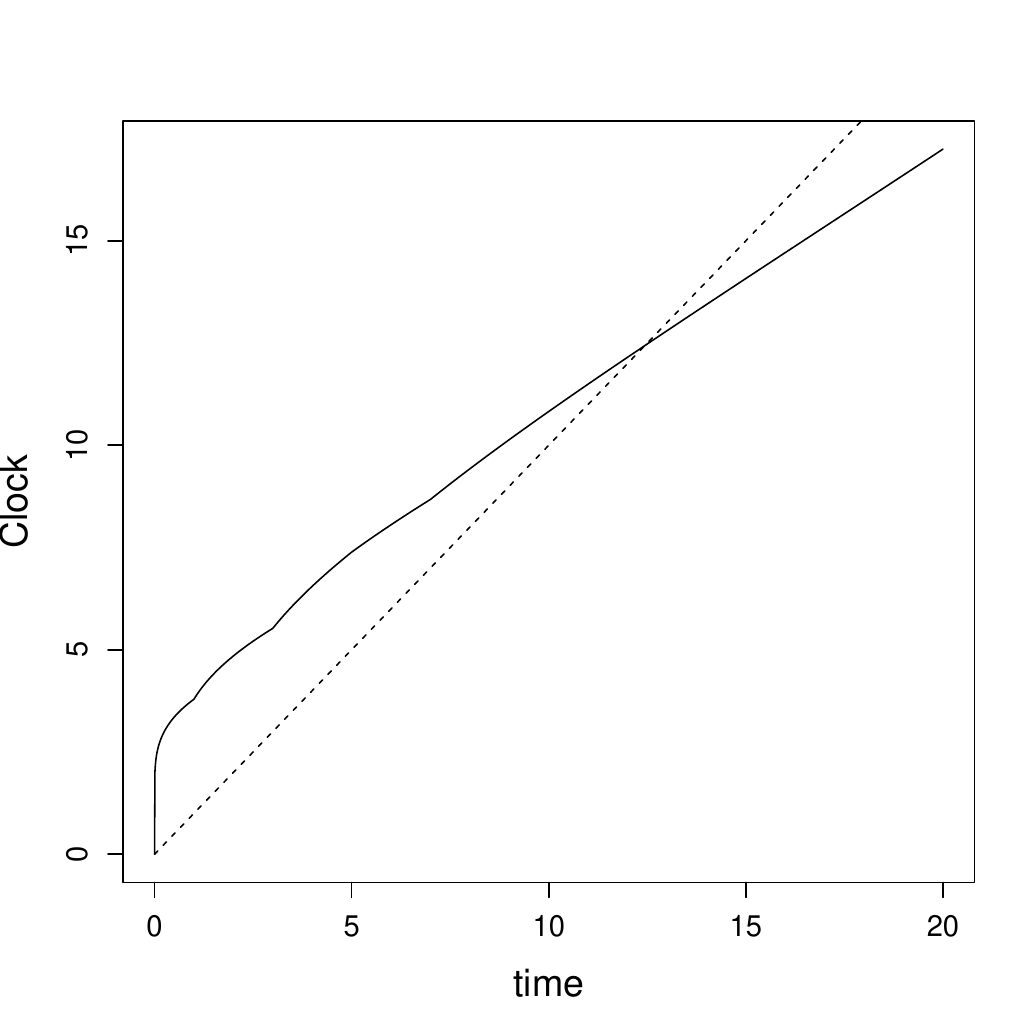}\label{fig:Clock-AT1P}}\\
\subfigure[AT1P model]{\includegraphics[width=0.32\columnwidth]{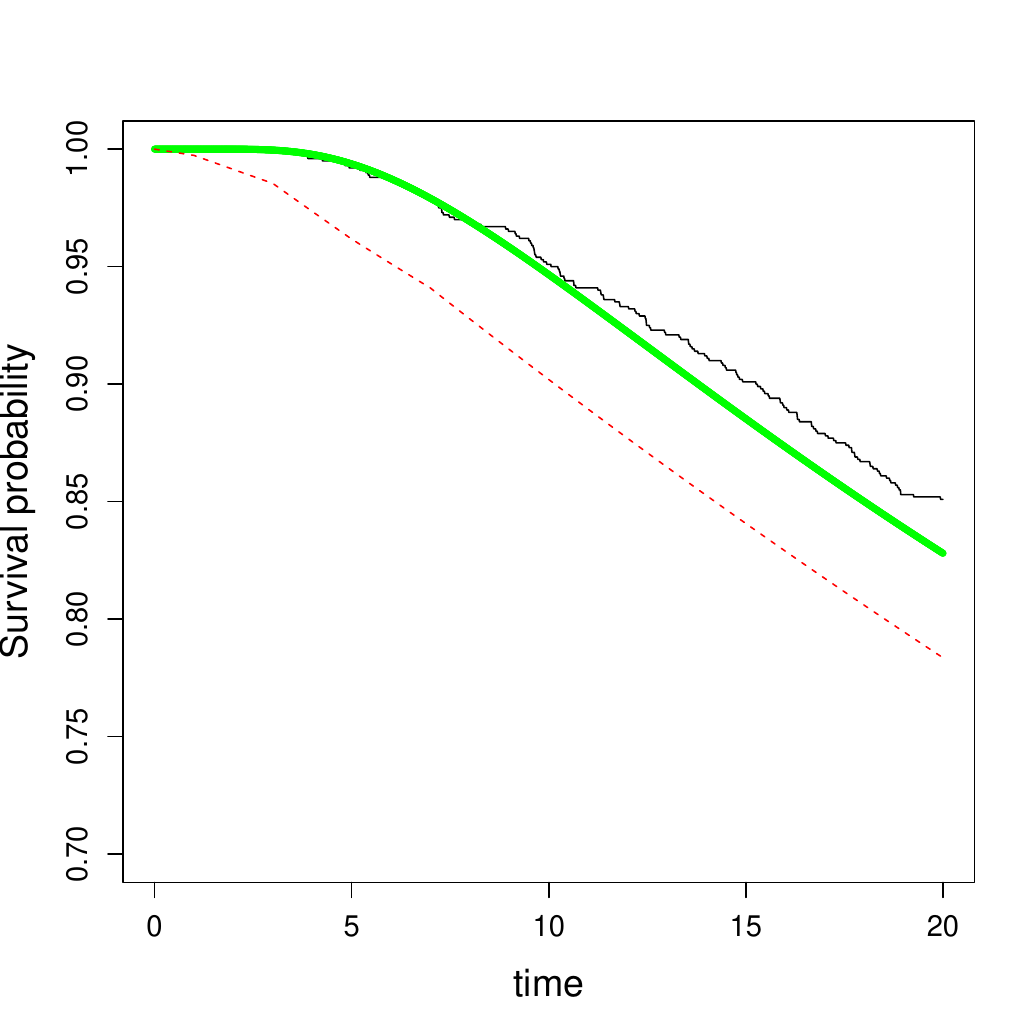}\label{fig:AT1P-B0d5}}
\subfigure[TC AT1P model]{\includegraphics[width=0.32\columnwidth]{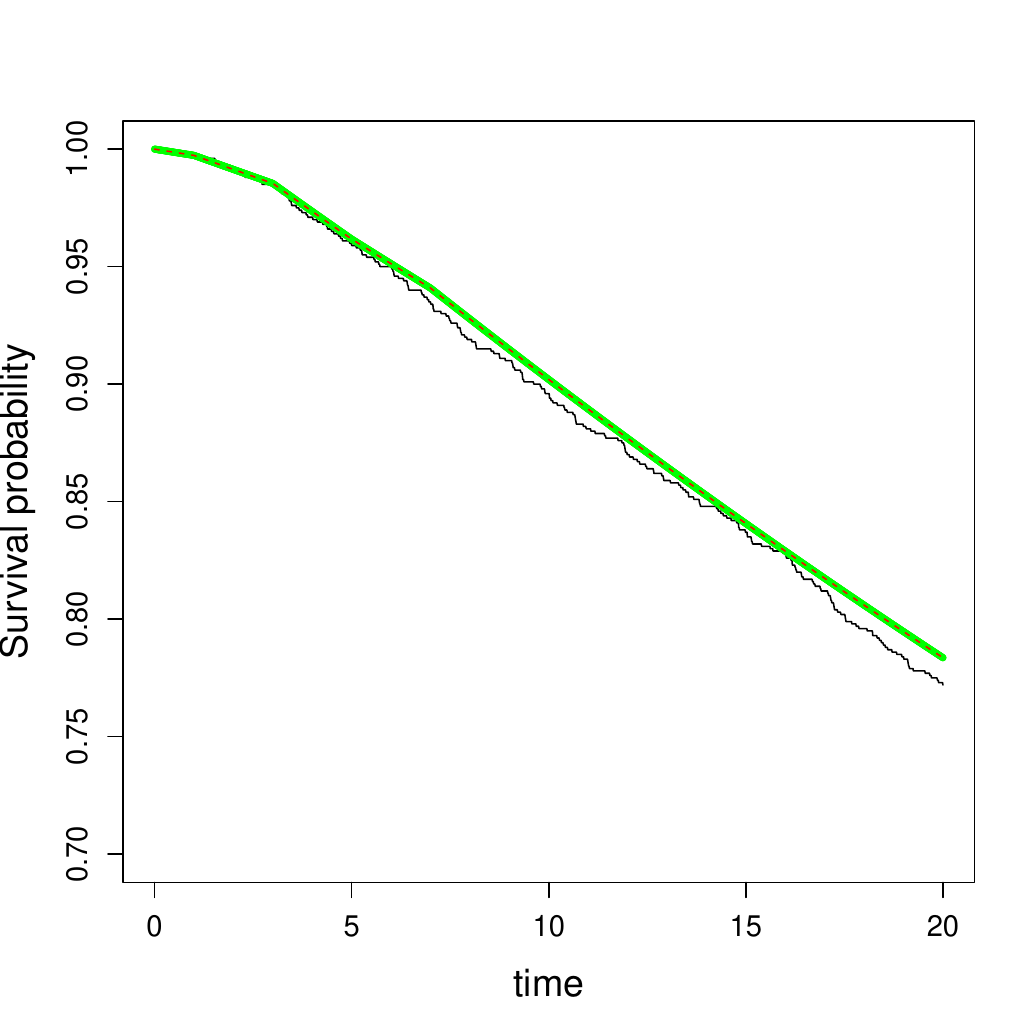}\label{fig:TC-AT1P-B0d5}}
\subfigure[Implied clock]{\includegraphics[width=0.32\columnwidth]{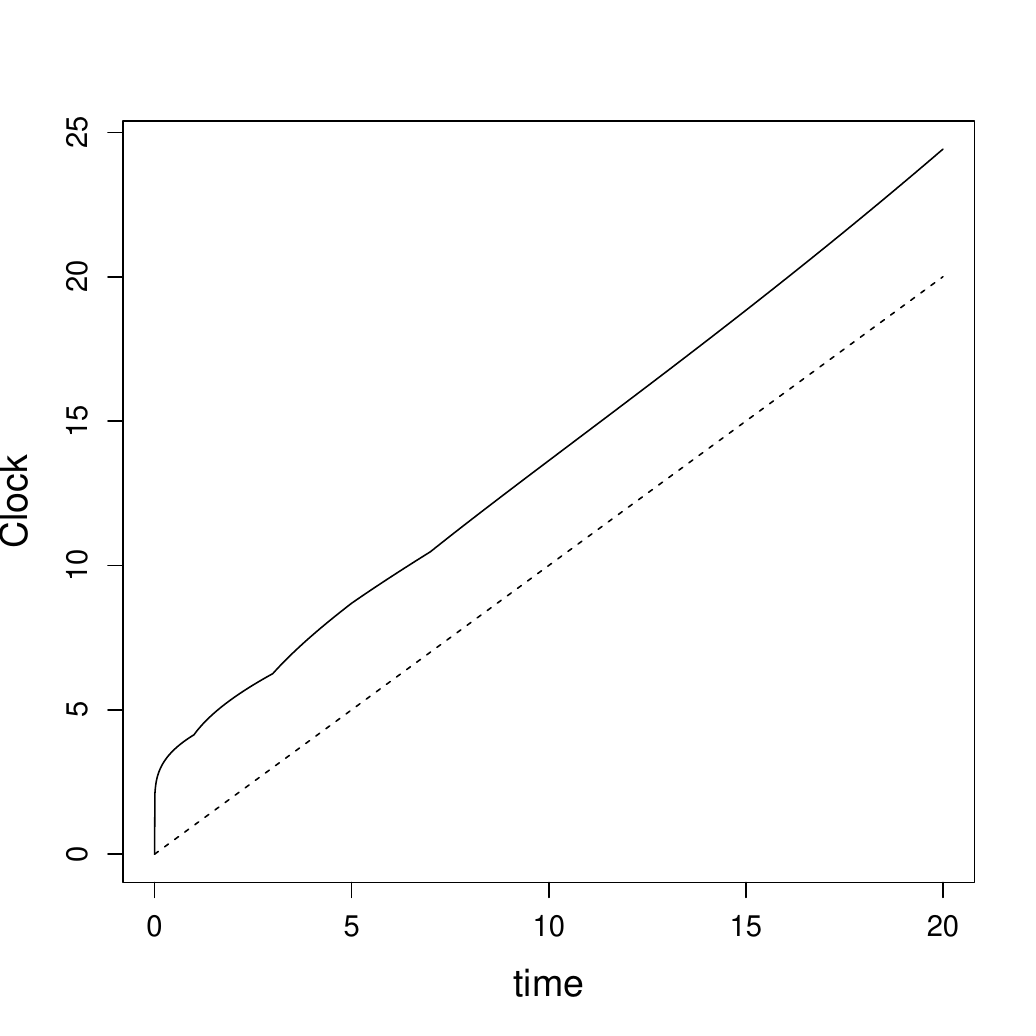}\label{fig:Clock-AT1P-B0d5}}
\caption{AT1P model and its perfect fit feature. Left: empirical (black) and theoretical (green) survival distributions of the AT1P, and the curve $G$ (red) built from the piecewise constant hazard rate function provided in Table 2 of \cite{brigomorinitarenghi}. Center: the empirical (black) and theoretical (green) survival distribution function of the default time in the time-changed AT1P with clock $\Theta=\Theta^{G,\tilde{W}}_k$, which now agree with $G$ (red). Right: implied clock. Parameters: $H=0.4, S_0=1,\sigma(u)=\sigma=0.15$ and $B=0$ (top) and $B=0.5$ (bottom).}\label{fig:AT1P}
\end{figure}

\section{Conclusion}
 
An important drawback of structural models is that most tractable specifications lack the flexibility to reproduce a prescribed survival probability curve from time $t=0$. We address this issue from a rather general perspective through deterministic time change: starting from a tractable first-passage model, we infer the clock for which the time-changed model has the specified default-time survival curve. We show that existence and uniqueness of the implied clock holds under mild conditions on the target curve and the latent model. The calibration is obtained by inversion of the latent survival curve. In the Brownian case, we recover the analytical expression provided in~\cite{OverbeckSchmidt2005}. For other tractable latent models, this can be efficiently done using numerical root-finding algorithms. In particular, the time-change representation clarifies the connection with the classic 2004 AT1P model \cite{Brigo2004,Brigo2005} and provides an exact calibration procedure for a regular target survival curve.
 
The construction also reproduces the short-end behavior encoded in the target survival curve. For continuous-path latent models, this may entail a clock rate, equivalently a variance rate, that becomes large \--- and actually explodes in standard settings\--- near the origin; nonetheless, the corresponding integrated variance remains finite on every finite horizon under the regularity conditions considered here.
 
Finally, calibration of the initial survival curve does not by itself determine all features relevant to a joint credit-and-market model. The clock impacts the instantaneous variance of the process and hence its cross-variation with other sources of risk. In counterparty-risk applications, for example, the clock or the associated volatility function may therefore affect the assessment of CVA under wrong-way risk. These considerations should be taken into account when specifying and using the model beyond the initial survival-curve calibration.


\ifdefined \MyBib
	\bibliography{\MyBib}
	\bibliographystyle{plain}
\fi

\section{Appendix}

\subsection{Local square-integrability of the diffusion coefficient in~\eqref{eq:sigma}}\label{app:sigmaL2}

\begin{theorem}\label{th:sigma}
Let $G$ be a regular survival function. Then the diffusion coefficient
defined in~\eqref{eq:sigma} satisfies
\[
\sigma\in L^2_{\mathrm{loc}}(\mathbb{R}_+),
\]
that is,
\[
\Sigma^2(T)=\int_0^T\sigma^2(t)\,dt<\infty,
\qquad T>0.
\]
\end{theorem}

\begin{proof}
Throughout the proof, set
\[
z(t):=\Phi^{-1}\left(\frac{1+G(t)}{2}\right),
\qquad
\Theta(t):=\left(\frac{k}{z(t)}\right)^2,
\qquad t>0,
\]
so that, by definition~\eqref{eq:sigma}, $\sigma^2(t)=k^2g(t)/\big(z(t)^3\phi(z(t))\big)$.
\medskip

\noindent\textit{Step 1: $z$ is absolutely continuous, and bounded away from $0$, on every compact subinterval of $(0,\infty)$.}
Fix $0<\varepsilon<T$. Since $G$ is a regular survival function (Definition~\ref{def:RegularSP}), $g=-G'>0$ almost everywhere, so $G$ is strictly decreasing on $[0,\infty)$; being also absolutely continuous and continuous, it maps $[\varepsilon,T]$ bijectively and monotonically onto the compact interval $[G(T),G(\varepsilon)]$, with
\[
0<G(T)<G(\varepsilon)<G(0)=1.
\]
Hence $\frac{1+G(t)}{2}$ ranges, for $t\in[\varepsilon,T]$, over the compact set $[a,b]:=\big[\tfrac{1+G(T)}2,\tfrac{1+G(\varepsilon)}2\big]\subset(0,1)$.
The function $\Phi^{-1}$ is continuously differentiable on $(0,1)$, with $(\Phi^{-1})'(u)=1/\phi(\Phi^{-1}(u))$; being continuous, this derivative is bounded on the compact set $[a,b]$, so $\Phi^{-1}$ is Lipschitz there. As a Lipschitz function composed with the absolutely continuous function $t\mapsto\frac{1+G(t)}2$, the function $z=\Phi^{-1}\circ\frac{1+G}2$ is itself absolutely continuous on $[\varepsilon,T]$. This shows that for all $t\in[\varepsilon,T]$, $z(t)$ belongs to the compact subset of $\big[\Phi^{-1}(a),\Phi^{-1}(b)\big]\subset(0,\infty)$
(since $a,b\in(\tfrac12,1)$ strictly).

\smallskip
\noindent\textit{Step 2: $\Theta$ is absolutely continuous on $[\varepsilon,T]$, with $\Theta'=\sigma^2$ there.}
Since $z$ takes values in the compact set $\big[\Phi^{-1}(a),\Phi^{-1}(b)\big]\subset(0,\infty)$ on $[\varepsilon,T]$, the map $x\mapsto k^2/x^2$ is Lipschitz on that set; composed with the absolutely continuous $z$, this shows $\Theta=k^2/z^2$ is absolutely continuous on $[\varepsilon,T]$. Differentiating $\Phi(z(t))=\frac{1+G(t)}2$ almost everywhere gives
\begin{equation}\label{eq:newproof1}
\phi(z(t))\,z'(t)=\frac{G'(t)}2=-\frac{g(t)}2,
\qquad\text{i.e.}\qquad
z'(t)=-\frac{g(t)}{2\phi(z(t))},
\qquad\text{a.e. }t\in[\varepsilon,T],
\end{equation}
leading to
\[
\Theta'(t)=-\frac{2k^2}{z(t)^3}\,z'(t)
=-\frac{2k^2}{z(t)^3}\left(-\frac{g(t)}{2\phi(z(t))}\right)
=\frac{k^2g(t)}{z(t)^3\phi(z(t))}
=\sigma^2(t),
\qquad\text{a.e. }t\in[\varepsilon,T].
\]

\smallskip
\noindent\textit{Step 3: Fundamental theorem of calculus on $[\varepsilon,T]$.}
Since $\Theta$ is absolutely continuous on $[\varepsilon,T]$,
\[
\int_\varepsilon^T\sigma^2(t)\,dt=\int_\varepsilon^T\Theta'(t)\,dt=\Theta(T)-\Theta(\varepsilon).
\]

\smallskip
\noindent\textit{Step 4: behaviour as $\varepsilon\downarrow0$, and passage to the limit.}
Because $G$ is regular, $G$ is continuous with $G(0)=1$, so $G(\varepsilon)\to1$ as $\varepsilon\downarrow0$, hence $\tfrac{1+G(\varepsilon)}2\to1$ and $z(\varepsilon)=\Phi^{-1}\big(\tfrac{1+G(\varepsilon)}2\big)\to+\infty$; consequently, $\lim_{\epsilon\downarrow 0}\Theta(\varepsilon)=0$.

The integrand $\sigma^2$ is nonnegative almost everywhere on $(0,T]$, and the intervals $[\varepsilon,T]$ increase to $(0,T]$ as $\varepsilon\downarrow0$. By the monotone convergence theorem applied to $\sigma^2\ind_{(0,T]}$,
\[
\int_0^T\sigma^2(t)\,dt
=\lim_{\varepsilon\downarrow0}\int_\varepsilon^T\sigma^2(t)\,dt
=\lim_{\varepsilon\downarrow0}\big[\Theta(T)-\Theta(\varepsilon)\big]
=\Theta(T).
\]
Finally, since $T>0$ and $G$ is regular, $G(T)\in(0,1)$ is an interior point of $(0,1)$, so $z(T)=\Phi^{-1}\big(\tfrac{1+G(T)}2\big)$ is a finite, strictly positive real number, and thus $\Theta(T)=k^2/z(T)^2<\infty$. Hence
\[
\Sigma^2(T)=\int_0^T\sigma^2(t)\,dt=\Theta(T)<\infty,
\qquad\forall T>0,
\]
which proves $\sigma\in L^2_{\mathrm{loc}}(\mathbb{R}_+)$.
\end{proof}

\begin{remark}\label{rem:sigmaL2-unconditional}
Two comments on Theorem~\ref{th:sigma} are worth making explicit.

First, the result holds for every regular survival function $G$ in the sense of Definition~\ref{def:RegularSP}, with no restriction on the local behaviour of the hazard rate $h$ near the origin: in particular, no growth condition is required. The possible blow-up of $\sigma^2(t)$ as $t\downarrow0$ described in Remark~\ref{rem:ThetaBM:Explosion} is therefore always compatible with $\sigma\in L^2_{\mathrm{loc}}(\mathbb{R}_+)$. 

Second, the function $\Theta$ constructed in the proof is not an auxiliary device introduced for convenience: comparing its definition with~\eqref{eq:ThetaCal} shows that it is the implied clock $\Theta_k^{G,W}$ of Theorem~\ref{th:clock}. The proof of Theorem~\ref{th:sigma} therefore also gives a direct, self-contained verification that $\Theta_k^{G,W}$ is absolutely continuous on every compact interval $[\varepsilon,T]\subset(0,\infty)$ and continuous at the origin with $\Theta_k^{G,W}(0)=0$. This complements Lemma~\ref{lem:iclock}. In the Brownian case treated here, Theorem~\ref{th:sigma} settles this point by an elementary, fully explicit argument (local Lipschitz composition on compact subintervals, plus monotone convergence at the origin).
\end{remark}

\subsection{Appendix: AT1P as a time-changed drifted Brownian motion}\label{app:AT1PasTC}

The trick consists in noting that the log firm-to-barrier ratio is of the form
\[
\mu\int_0^t\sigma^2(u)\,du+\int_0^t\sigma(u)\,dW_u
=\mu\Theta(t)+W_{\Theta(t)},
\qquad
\Theta(t):=\Sigma^2(t)=\int_0^t\sigma^2(u)\,du.
\]
Thus $\Theta$ is continuous and non-decreasing, which is sufficient for the
following change of time range. If $\Theta$ is to be called a clock in the
sense of Definition~\ref{def:clock}, additionally assume $\sigma^2>0$ almost everywhere
and $\Sigma^2(t)\to\infty$ as $t\to\infty$.
\medskip 

In AT1P, it holds that $S_0>H$ such that $\log (H/S_0)<0$. Setting $\mu=B-1/2$ and using the explicit solution of the process $S$, one gets
\begin{align*}
G^{\tiny\text{AT1P}}(t)
&=\Pr\left(\inf_{0\le s\le t}
\left\{\mu\Sigma^2(s)+W_{\Sigma^2(s)}\right\} >\log\frac{H}{S_0}\right)\\
&=\Pr\left(\inf_{0\le u\le\Sigma^2(t)}
\left\{\mu u+W_u\right\}>\log\frac{H}{S_0}\right)\\
&=\Phi\left(\frac{\mu\Sigma^2(t)+\log(S_0/H)}{\Sigma(t)}\right)
-\left(\frac{H}{S_0}\right)^{2\mu}
\Phi\left(\frac{\mu\Sigma^2(t)-\log(S_0/H)}{\Sigma(t)}\right),
\end{align*}
which is exactly~\eqref{eq:G-AT1P}. Here, the second equality uses that
$\Sigma^2$ is continuous and non-decreasing, so that it maps $[0,t]$ onto
$[0,\Sigma^2(t)]$.

\subsection{Calibration of the AT1P under the piecewise constant volatility parametrization}\label{app:CalAT1P}

The previous discussion shows that, if the volatility function in AT1P is left free, then the model can be calibrated exactly to any regular target survival curve $G$ by imposing
\[
\Sigma^2(t)=\Theta^{G,\tilde{W}}_{\log(K_0/S_0)}(t).
\]
In practice, however, the AT1P model is often parametrized with a volatility function that is piecewise constant on a prescribed grid. The next result shows that this restriction still allows one to perfectly reproduce a set of probabilities computed from a regular survival curve at prescribed times.

\begin{lemma}[Grid-point fit of $G$ using AT1P with pwc volatility]\label{th:pwcAT1P}
Consider the grid of times
\[
0=t_0<t_1<\ldots<t_n<\infty
\]
and let $G$ be a regular survival curve. Then, there exists an AT1P model whose volatility function is piecewise constant on each interval $[t_{i-1},t_i)$, delivering a perfect fit to $G$ on the grid, i.e.,
\[
G^{AT1P,pwc}(t_i)=G(t_i),\qquad i=1,\ldots,n.
\]
\end{lemma}

\begin{proof}
Fix the AT1P parameters $(K_0,B)$ such that
\[
k=\log(K_0/S_0)<0,\qquad \mu=B-\frac{1}{2}\leq 0.
\]
By the discussion in Section~4.2, the AT1P survival curve can be written as
\[
G^{AT1P}(t)=G^{\tilde{W}}_k(\Sigma^2(t)),\qquad t\geq 0,
\]
where $\tilde{W}_t=\mu t+W_t$ and $\Sigma^2(t)=\int_0^t \sigma^2(u)\,du$ matches the implied clock associated with the target curve $G$, $\Sigma^2(t)=\Theta^{G,\tilde{W}}_k(t)$. Define now another clock $\Theta$ by linear interpolation between the values of the implied clock at the grid points:
\[
\Theta(t_i)=\Theta^{G,\tilde{W}}_k(t_i),\qquad i=0,\ldots,n,
\]
with $\Theta(0)=0$. 

On $[t_n,\infty)$, extend $\Theta$ with any strictly positive constant rate, for example set
\[ \Theta(t)=\Theta(t_n)+(t-t_n),\qquad t\ge t_n \]
Then $\Theta$ is a clock on $[0,\infty)$ in the sense of
Definition~\ref{def:clock}. On every calibration interval
$[t_{i-1},t_i)$ its rate is the positive constant
\[\theta_i= \frac{\Theta_k^{G,\tilde W}(t_i)-\Theta_k^{G,\tilde W}(t_{i-1})}{t_i-t_{i-1}},\qquad i=1,\ldots,n.\]
Taking $\sigma_i=\sqrt{\theta_i}$ gives a strictly positive piecewise constant AT1P volatility.

On $[0,t_n]$, take $\sigma(t)=\sqrt{\theta_i}$ on each interval
$[t_{i-1},t_i)$. On $(t_n,\infty)$, take, for example,
$\sigma(t)=1$. Then
\[
\Theta(t)=\int_0^t\sigma^2(u)\,du,
\qquad 0\le t\le t_n,
\]
and the corresponding AT1P model satisfies
\[
G^{\mathrm{AT1P,pwc}}(t)=G_k^{\widetilde W}(\Theta(t)).
\]

This shows that there exists a volatility function $\sigma$, piecewise constant on each interval $[t_{i-1},t_i)$, such that
\[
\Theta(t)=\int_0^t \sigma^2(u)\,du.
\]

Consider the corresponding AT1P model. Its survival curve satisfies
\[
G^{AT1P,pwc}(t)=G^{\tilde{W}}_k(\Theta(t)).
\]
Therefore, at each grid point $t_i$,
\[
G^{AT1P,pwc}(t_i)
=
G^{\tilde{W}}_k(\Theta(t_i))
=
G^{\tilde{W}}_k\!\left(\Theta^{G,\tilde{W}}_k(t_i)\right)
=
G(t_i),
\]
which proves the claim.
\end{proof}

\begin{remark}
Lemma~\ref{th:pwcAT1P} is a statement about matching a finite set of survival probabilities on a prescribed grid. This is closely related to, but distinct from, exact calibration to quoted CDS par spreads under a given market pricer. The latter depends on the precise pricing convention, including the dates entering the pricer and the interpolation/bootstrapping convention used to extract the target curve. We do not pursue this point here. Empirically, however, it is worth noting that structural calibrations such as AT1P often remain admissible in stressed situations where simple reduced-form hazard parametrizations may become unstable; in the present framework, this corresponds to an extreme but still admissible deterministic clock / variance profile.
\end{remark}

\end{document}